\documentclass[12pt]{article}

\pdfoutput=1

\usepackage[margin=1in]{geometry}

\usepackage[
  bookmarks=true,
  bookmarksnumbered=true,
  bookmarksopen=true,
  pdfborder={0 0 0},
  breaklinks=true,
  colorlinks=true,
  linkcolor=black,
  citecolor=black,
  filecolor=black,
  urlcolor=black,
]{hyperref}
\usepackage[round]{natbib}
\usepackage{amssymb,amsmath,amsthm}
\usepackage{enumitem}
\usepackage[capitalise]{cleveref}
\usepackage{nicefrac}
\usepackage{isomath}
\usepackage{mathtools}
\usepackage[boxed]{algorithm2e}

\usepackage{fnbreak}

\allowdisplaybreaks

\theoremstyle{plain}
\newtheorem{theorem}{Theorem}
\newtheorem{corollary}{Corollary}
\newtheorem{lemma}{Lemma}[section]
\newtheorem{open-problem}{Open Problem}
\Crefname{algocf}{Algorithm}{Algorithms}
\Crefname{equation}{Equation}{Equations} 

\theoremstyle{definition}
\newtheorem*{definition}{Definition}

\newcommand{\eqdef}{\triangleq}

\newcommand{\expect}[2]{\mathbb{E}_{#1}\left[#2\right]}
\DeclareMathOperator{\OPT}{OPT}
\DeclareMathOperator{\Rev}{Rev}

\newcommand{\epsfloor}[1]{\lfloor#1\rfloor_{\varepsilon}}

\DeclareMathOperator{\poly}{poly}
\DeclareMathOperator{\supp}{supp}

\hypersetup{
  pdfauthor      = {Yannai A. Gonczarowski and S. Matthew Weinberg},
  pdftitle       = {The Sample Complexity of Up-to-Epsilon Multi-Dimensional Revenue Maximization},
}

\title{The Sample Complexity of\texorpdfstring{\\}{ }Up-to-\texorpdfstring{$\varepsilon$}{Epsilon} Multi-Dimensional Revenue Maximization}
\author{Yannai A. Gonczarowski\thanks{Microsoft Research. \emph{E-mail}: \href{mailto:yannai@gonch.name}{yannai@gonch.name}. Research conducted while at the Hebrew University of Jerusalem and Microsoft Research.} \and S. Matthew Weinberg\thanks{Department of Computer Science, Princeton University. \emph{E-mail}: \href{mailto:smweinberg@princeton.edu}{smweinberg@princeton.edu}.}}
\date{March 2, 2021}

\begin{document}

\maketitle

\begin{abstract}
We consider the sample complexity of revenue maximization for multiple bidders in unrestricted multi-dimensional settings. Specifically, we study the standard model of $n$ additive bidders whose values for $m$ heterogeneous items are drawn independently. For any such instance and any $\varepsilon > 0$, we show that it is possible to learn an $\varepsilon$-Bayesian Incentive Compatible auction whose expected revenue is within $\varepsilon$ of the optimal $\varepsilon$-BIC auction from only polynomially many samples. 

Our fully nonparametric approach is based on
ideas that hold quite generally, and completely sidestep the difficulty of characterizing optimal (or near-optimal) auctions for these settings. Therefore, our results easily extend to general multi-dimensional settings, including valuations that are not necessarily even \emph{subadditive}, and arbitrary allocation constraints.
For the cases of a single bidder and many goods, or a single parameter (good) and many bidders, our analysis yields exact incentive compatibility (and for the latter also computational efficiency). Although the single-parameter case is already well-understood, our corollary for this case extends slightly the state-of-the-art.
\end{abstract}

\section{Introduction}

A fundamental question at the heart of the literature on mechanism design is that of \emph{revenue maximization} by a single seller who is offering for sale any number of goods to any number of (potential) bidders. In the classic economic literature, this problem is studied in a \emph{Bayesian} setting: the seller has prior knowledge of (often, independent) distributions from which the valuation of each bidder for each good is drawn, and wishes to devise a truthful mechanism that maximizes her revenue in expectation over these prior distributions. Over the past few years, numerous works at the interface of economics and computation are now studying a more demanding model: that of \emph{mechanism design from samples}. In this model, rather than possessing complete knowledge of the distributions from which the bidders' values for the various items are drawn, the seller more realistically only has access to samples from these distributions (e.g., past data). The goal in this setting is to learn with high probability an auction with good revenue guarantees given polynomially many (in the parameters of the problem) samples.

Revenue maximization from samples is commonly thought of as a ``next step'' beyond Bayesian revenue maximization. That is, existing works so far in this context take settings for which simple auctions in the related Bayesian problem are already well-understood and prove that these simple auctions can be learned efficiently via samples (up to an $\varepsilon$ loss, which will always be lost when optimizing from samples). For example: in single-parameter settings, seminal work of~\citet{m81} completely characterizes a simple and optimal auction in the Bayesian setting, and works such as~\citet{cr14, mr15, dhp16, ht16, rs16, gn17} prove that these simple mechanisms or variants thereof can be learned with polynomially many samples. Similarly, in multi-parameter settings with independent items, works of~\citet{chk07,chms10,cms15,hn12,bilw14,rw15,y15,cdw16,cm16,cz17} prove that simple mechanisms achieve constant-factor approximations in rich multi-dimensional settings, and works of~\citet{mr16,bsv16,bsv17,cd17,s17} prove that simple mechanisms with these guarantees can be learned with polynomially many samples. These analyses rely on a delicate understanding of the structure and/or inherent dimensionality of auctions that give such revenue guarantees to show how to learn such an auction without overfitting the samples. 

It is therefore unsurprising that the problem of learning an up-to-$\varepsilon$ revenue-maximizing multi-item auction from samples has not been previously studied, since the structure/\linebreak dimensionality of optimal (precisely or up-to-$\varepsilon$) multi-item auctions is not understood even when there is only one bidder, and even with independent items. Such auctions are known to be extremely complex, suffering from properties such as randomization~\citep{t04}, uncountable menu complexity~\citep{ddt13}, and non-monotonicity~\citep{hr15}. These domains provably lack the natural starting point of all previous works: a structured/low-dimensional mechanism in the Bayesian setting to learn via samples. 

In this paper we show that despite these challenges, up-to-$\varepsilon$ optimal multi-item auctions can be learned from polynomially many samples from the underlying bidder-item distributions. More formally, in a setting with $n$~bidders and $m$ items where the value of each bidder~$i$ for each item $j$ is drawn independently from a distribution $V_{i,j}$ supported on $[0,H]$ for some $H$ that is known to the seller, we show that polynomially many samples suffice for learning, with probability at least $1-\delta$, an $m$-item almost-truthful auction that maximizes the expected revenue among all possible $m$-item almost-truthful auctions up to an additive~$\varepsilon$. Below, BIC refers to \emph{Bayesian Incentive Compatible}: an auction for which it is in every bidders' interest to bid truthfully, given that all other bidders do so as well.

\begin{theorem}[Main Result --- informal version of \cref{polynomial-bic}]\label{intro-bic}
For $n$ bidders with independent values for $m$ items supported on $[0,H]$, for every $\varepsilon,\delta>0$ and for every $\eta\le\poly(n,m,H,\varepsilon)$, the sample complexity of learning, w.p.\ $1\!-\!\delta$, an $\eta$-BIC auction that maximizes revenue (among all $\eta$-BIC auctions) up to an additive $\varepsilon$ is $\poly(n,m,H,\nicefrac{1}{\varepsilon},\nicefrac{1}{\eta},\log\nicefrac{1}{\delta})$.
\end{theorem}

The above \lcnamecref{intro-bic} is informal mostly because we have not specified exactly how bidders value bundles of items. Essentially the bidders may have arbitrary (i.e., not necessarily additive, not even necessarily subadditive) valuations subject to some Lipschitz condition (i.e., changing the value of bidder $i$ for item $j$ by $\varepsilon$ only changes the bidder's value for any outcome by at most $L\varepsilon$ for some absolute constant $L$).\footnote{Our results in fact hold even more generally: to arbitrary outcomes that do not even correspond to bundles of items. See \cref{model} for the full details.} We defer a formal definition to Section~\ref{model}, but only note here that commonly studied classes of valuations such as additive, unit-demand, or ``additive subject to constraints'' with independent items (as well as several natural subadditive and superadditive valuation classes, and even certain valuation classes with externalities) all satisfy our definition with Lipschitz constant $1$.

The main challenge in proving our result for $m>1$ items is noted above: the structure of (up-to-$\varepsilon$) optimal mechanisms for such settings is not understood, even for additive valuations. In particular, there is no known low-dimensional class of mechanisms that is guaranteed to contain an (up-to-$\varepsilon$) optimal mechanism for any product distribution, thus barring the use of many learning-theoretic arguments. Our result relies on a succinct structured argument, allowing us to reduce revenue maximization from samples to related problems of revenue maximization from given discrete distributions.

As the corresponding Bayesian question remains open (i.e., whether one can find, given the distributions explicitly, an up-to-$\varepsilon$ optimal mechanism in poly-time), our result is of course information-theoretic: it shows that polynomially many samples suffice for a computationally unbounded seller, but provides no computationally efficient learning algorithm. Concretely, the algorithm that we give uses as a black box an oracle that can perform (optimal or almost-optimal) multi-item Bayesian revenue maximization given (the full description of) finite prior distributions.\footnote{Note however that if computationally efficient algorithms were to be developed for up-to-$\varepsilon$ optimal mechanisms given an explicit prior, then our approach would immediately become computationally efficient as well.}

\subsection{Brief Overview of Techniques}
Most prior works (for single- as well as multi-dimensional settings) take the following approach: first, define a class $\mathcal{C}_\varepsilon$ of auctions as a function of $\varepsilon$. Second, prove that, for all possible distributions $\mathcal{D}$, the class $\mathcal{C}_\varepsilon$ contains an up-to-$\varepsilon$ optimal mechanism for $\mathcal{D}$. Finally, prove that the best-in-class (up to $\varepsilon$) for $\mathcal{C}_\varepsilon$ can be learned with polynomially many samples. In prior works, ingenuity is required for both steps: $\mathcal{C}_\varepsilon$ is explicitly defined, proved to contain up-to-$\varepsilon$ optimal auctions, and proved to have some low-dimensional structure allowing efficient learnability.

Our approach indeed follows this rough outline, with two notable simplifying exceptions. First is our approach to defining $\mathcal{C}_\varepsilon$. Here, we first define $\mathcal{C}_{\varepsilon,S}$ be the space of all auctions that are optimal for an empirical distribution over $S$-many $\poly(\varepsilon)$-rounded samples (that is, optimal for any discrete product distribution where each marginal is:\ \ a)~only supported on multiples of $\poly(\varepsilon)$ and\ \ b)~uniform over a multiset of size $S$). 
This fully nonparametric approach stands in contrast with the popular existing approach of taking $\mathcal{C}_\varepsilon$ to be some parametric family of auctions.
Unlike in such existing approaches where $C_{\varepsilon}$ is fixed, in our approach the set $C_{\varepsilon,S}$ grows with the number of samples $S$. Nonetheless, we show that the rate of its growth is moderate enough so that there exists a ``sweet-spot'' number of samples $S=\poly(\nicefrac{1}{\varepsilon})$ such that on the one hand $C_{\varepsilon,S}$ contains an auction that is up-to-$\varepsilon$ optimal for the ``true distribution'' $\mathcal{D}$ and on the other hand, the best-in-class from $C_{\varepsilon,S}$ can be learned from $S$ samples. So in the language of prior work, one could say that we set $\mathcal{C}_{\varepsilon}=\mathcal{C}_{\varepsilon,S}$ for this $S=\poly(\nicefrac{1}{\varepsilon})$.

To show that $\mathcal{C}_\varepsilon$ does in fact contain, for all distributions $\mathcal{D}$, an auction that is up-to-$\varepsilon$ optimal for $\mathcal{D}$, we simply take enough samples to guarantee uniform convergence (of the revenue) over $\mathcal{C}_{\varepsilon}$ \emph{and additionally} the optimal auction for $\mathcal{D}$. It's far from obvious why this should suffice, as the optimal auction for $\mathcal{D}$ is not an element of $\mathcal{C}_{\varepsilon}$, nor even of the same format.\footnote{That is, the optimal auction for $\mathcal{D}$ is a mapping from the the support of $\mathcal{D}$ to outcomes, whereas the elements of $\mathcal{C}_{\varepsilon}$ are mappings from a finite space to outcomes. Furthermore, notions of Bayesian incentive compatibility for $\mathcal{D}$ do not imply nor are implied by these notions for the various discrete distributions defining $\mathcal{C}_{\varepsilon}$.} Still existing tools (namely, the $\varepsilon$-BIC-to-BIC reduction of~\citet{dw12, rw15}), when applied correctly, suffice to complete the argument. This part of our proof is conceptually much simpler than prior works (despite making use of a big technical hammer), as this approach holds quite generally and is robust due to not requiring the analysis of any specific class of mechanisms.

Second, our argument that the best-in-class can be in fact learned (up to $\varepsilon$) with $S=\poly(\varepsilon)$ samples is simply a counting argument, and does not require any notions of a learning dimension. This is indeed in the spirit of some recent \emph{single-dimensional results}, however in those results the counting argument is highly dependent on the structure of auctions in $C_{\varepsilon}$. As discussed above, such dependence is damning for multi-dimensional settings where such structure provably doesn't exist. Again, the proof does require some hammers (notably, arguments originally developed for \emph{reduced forms} via samples in~\citet{cdw12}, and a concentration inequality of~\citet{bbp17,dhp16}), but they are applied in a fairly transparent manner.

The above approach should help explain how we are able to extend far beyond prior works, which relied on a detailed analysis of specific structured mechanisms: The key tools we use are applicable quite generally, whereas the specific mechanisms analyzed in prior work are only known to maintain guarantees only in restricted settings. For example, \cref{intro-bic} already constitutes the first up-to-$\varepsilon$ optimal-mechanism learning result for any multi-parameter setting \emph{even if it held only for additive valuations (and one bidder)}. But the approach is so general that extending it to arbitrary Lipschitz valuations with independent items is simply a matter of updating notation.

\subsection{Applications and Extensions}
Specialized to a single-bidder setting, our construction in fact yields exact truthfulness (more on that in \cref{exact}), showing that an $\varepsilon$-optimal mechanism can be found for a single bidder with independent item values (with Lipschitz valuations) using only polynomially many samples. This should be contrasted with a result of \citet{dhn14}, which shows that achieving this is \emph{not} possible for correlated distributions, even for a bidder with additive valuations.

\begin{corollary}[Single Bidder --- informal version of \cref{polynomial-single-bidder}]\label{intro-single-bidder}
For one bidder with independent values for $m$ items supported on $[0,H]$, for every $\varepsilon,\delta>0$, the sample complexity of learning, w.p.\ ${1\!-\!\delta}$, an IC auction that maximizes revenue (among all IC auctions) up to an additive $\varepsilon$ is $\poly(m,H,\nicefrac{1}{\varepsilon},\log\nicefrac{1}{\delta})$.
\end{corollary}

Specialized to single-dimensional settings, our analysis once again yields a strengthened result, both in giving exact Dominant Strategy Incentive Compatibility (DSIC)\footnote{A mechanism is DSIC if it is a dominant strategy for each bidder to bid truthfully. For a single good, \citet{m81} shows that the maximal revenue attainable by a BIC mechanism and by a DSIC mechanism is the same.}, and in providing a computationally efficient algorithm (due to known efficient solutions \citep{e07,m81} to single-parameter revenue maximization from given discrete distributions):

\begin{corollary}[Single-Parameter --- informal version of \cref{polynomial-single-parameter}]\label{intro-single-parameter}
For $n$ single-parameter bidders with independent values in $[0,H]$, for every $\varepsilon,\delta>0$, the sample complexity of \mbox{\textbf{efficiently}} learning, w.p.\ $1\!-\!\delta$, a DSIC auction that maximizes revenue (among all DSIC auctions) up to an additive $\varepsilon$ is $\poly(n,H,\nicefrac{1}{\varepsilon},\log\nicefrac{1}{\delta})$.
\end{corollary}

\cref{intro-single-parameter} nicely complements the existing literature on single-parameter sample complexity in the following ways. First, our algorithm/analysis immediately follows as a special case of \cref{intro-bic} (without referencing structural results about optimal single-parameter auctions), so it is in some sense more principled. Second, our analysis holds even for arbitrary constraints on the allocations (putting it in the same class as the state-of-the-art\footnote{Here and throughout the paper when we refer to ``state-of-the-art'' for single-parameter settings, we are specifically referring to allocation constraints that can be accommodated.} single-parameter results \citep{ht16, gn17}, and even slightly beyond\footnote{Note that for the single-parameter setting, our algorithm in fact coincides with that of \citet{dhp16}. However our analysis, unlike theirs, extends to arbitrary allocation constraints. Our approach also transparently handles mild extensions of constraints \emph{beyond} those considered in~\citet{ht16,gn17}. \citet{gn17} explicitly state that their techniques cannot handle such extensions and leave this question (which we successfully resolve) open.}).

Finally, portions of our approach are specific to Bayesian Incentive Compatible auctions (versus Dominant Strategy Incentive Compatible auctions), but portions are not. We're therefore able to use the same techniques to conclude similar, albeit qualitatively weaker, results for $\varepsilon$-DSIC auctions in~\cref{polynomial-dsic}. See~\cref{dsic} for further details.

\subsection{Related Work and Brief Discussion}
Two active lines of work are directly related to the present paper. First are papers that study rich multi-dimensional settings, and aim to show that mechanisms with good approximation guarantees can be learned with few samples, such as~\citet{mr16, bsv16, bsv17, cd17, s17}. The main approach in each of these works is to show that specific classes of structured mechanisms (e.g., classes that are known to allow for constant-factor revenue maximization) are inherently low-dimensional with respect to some notion of dimensionality. Our results are stronger than these in some regards and weaker in others. More specifically, our results are stronger in the sense that with comparably many samples, our mechanisms guarantee an up-to-$\varepsilon$ approximation to the optimal mechanism instead of a constant-factor. Our results are weaker in the sense that our learning algorithms are information-theoretic (do not run in poly-time), and our mechanisms are not ``as simple.'' As discussed earlier, both weaknesses are necessary in order to possibly surpass the constant-factor barrier (at least, barring the resolution of major open questions, such as a computationally efficient up-to-$\varepsilon$ approximation even when all distributions are explicitly known. Again, note that should this question be resolved affirmatively, our results would immediately become computationally efficient as well). 

Most related to our work, at least in terms of techniques, is the rich line of works on \emph{single-dimensional settings}~\citep{dry10, cr14, hmr15, mr15, dhp16, ht16, rs16, gn17}. These works show that up-to-$\varepsilon$ optimal mechanisms can be learned in richer and richer settings. In comparison to these works, our single-dimensional results slightly extend the state-of-the-art \citep{ht16, gn17} \emph{as a corollary} of a more general theorem that applies to multi-dimensional settings. Even restricted to single-dimensional settings, our proof is perhaps more transparent. 

We conclude with a brief discussion and an open problem. \cref{intro-single-bidder,intro-single-parameter} are both deduced from \cref{intro-bic} by use of an argument as to why the resulting $\varepsilon$-BIC auction is in fact BIC, or by using an $\varepsilon$-BIC to BIC reduction that loses negligible revenue. Given that we have explicitly referenced the existence of a quite general $\varepsilon$-BIC-to-BIC reduction, the reader may be wondering why this reduction does not in fact allow our general results to be exactly BIC as well.

The main barrier is the following: in order to actually \emph{run} the $\varepsilon$-BIC-to-BIC reduction as part of our auction for $n > 1$ bidders, one must take samples \emph{exponential in the number of items} from each bidders' value distribution. This means that even though we can learn an $\varepsilon$-BIC mechanism with few samples, plugging it through the reduction to remove the $\varepsilon$ would cost us exponentially many samples in addition. Note that our current use of these theorems is non-constructive: we only use them to claim that the revenues achievable by the optimal BIC and $\varepsilon$-BIC mechanisms are not far off. This conclusion does not actually require \emph{running} the reduction, but rather simply observing that it could be run (more details in \cref{bic}).  

When bidder valuations are drawn from a product distribution, it seems conceivable (especially given our results), that sample complexity polynomial in the number of items should suffice. Indeed, if each bidders' values are drawn i.i.d., this is known due to exploitations of symmetry~\citep{dw12}. But subexponential sample complexity is not known to suffice for any other restricted class of distributions, despite remarkable recent progress in developing connections to combinatorial Bernoulli factories~\citep{dhkn17}. We state below what we consider to be the main open problem left by our work in the context of this paper, but readers familiar with black-box reductions in Bayesian mechanism design will immediately recognize a corresponding open problem for the original welfare-maximization setting studied in~\citet{hkm11, bh11} that is equally enticing.

\begin{open-problem}\label{eps-bic-bic-samples}
Given an $\varepsilon$-BIC auction for some product distribution over $n\!>\!1$ bidders and $m\!>\!1$ items, even with additive bidders, is it possible to transform it into a (precisely) BIC auction with negligible ($\poly(\varepsilon)\!\cdot\!\poly(n,m,H)$) revenue loss using polynomially many samples from this product distribution?
\end{open-problem}

\subsubsection{Subsequent Related Work}
Since the this work came out in FOCS 2018 \citep{gw18}, progress has been made on three fronts which merits brief discussion. First is recent work of \citet{KothariMSSW19}, which gives a QPTAS for a single unit-demand buyer over $m$ independent items. As referenced in the previous section, our results now immediately extend to provide an up-to-$\varepsilon$ approximation to the optimal revenue for a single unit-demand buyer using polynomially-many samples with quasi-polynomial runtime. 

Second is recent work of \citet{GergatsouliLT19}, which identifies barriers towards a problem related to \cref{eps-bic-bic-samples}. Specifically, one can view \cref{eps-bic-bic-samples} as asking for an algorithm that takes as input a \emph{fully-known auction}, and has \emph{sample access to $D$}, and outputs a BIC auction. \citet{GergatsouliLT19} prove an impossibility result in a setting where the algorithm is given as input a \emph{fully-known $D$}, but only has \emph{query access to the auction} (there are also a few other additional differences). To best see why this difference is significant, consider a single bidder and $m$ items. With a (fully-known) $\varepsilon^2$-IC auction $A$, multiplying all prices by a multiplicative $(1-\varepsilon)$ results in a IC auction that up-to-$\varepsilon$ approximates the revenue of $A$, using no samples from $D$ (i.e., this resolves the single-bidder version of  \cref{eps-bic-bic-samples}; see \cref{exact-one-bidder}). However, the auction~$A$ may offer $\text{exp}(m)$ different outcomes, so actually learning $A$ via query access to figure out which outcome of $A$ a given buyer might choose (even when $D$ is explicitly known) might require $\text{exp}(m)$ queries to $A$. Indeed, this reduction positively resolves the single-bidder version of \cref{eps-bic-bic-samples}, but breaks down for the question studied by \citet{GergatsouliLT19}, which they instead \emph{negatively} resolve. \citet{GergatsouliLT19} therefore identify barriers to a positive resolution of a problem related to \cref{eps-bic-bic-samples}, but in a model that is different enough that one should not make too strong a conjecture regarding \cref{eps-bic-bic-samples} based on their results.

Finally is recent work of~\citet{GuoHTZ19}, which applies techniques similar to ours to study the precise sample complexity of several domains, including single-buyer revenue maximization (the scenario we discuss in~\cref{exact-one-bidder}). While we are most interested in establishing that the sample complexity is polynomial (rather than nailing it down precisely), the latter work establishes that certain parts of our analysis are in fact tight up to logarithmic factors. Still, it is not known whether our full results are tight, and it is a fascinating open problem to determine this (also posed explicitly in~\citet{GuoHTZ19}). Specifically,~\citet{GuoHTZ19} also establish a lower bound on learning an up-to-$\varepsilon$ optimal hypothesis over a product distribution, as a function of the size of support of the marginals. Put another way, their lower bound implies that improving our sample complexity (beyond logarithmic factors) would require either\ \ a) discretizing differently, and/or quantitatively improving state-of-the-art $\varepsilon$-IC to IC reductions, or\ \ b) exploiting some (currently unknown) structure of revenue-maximizing multi-item auctions.\footnote{As already noted, understanding the structure of such auctions is a major open problem independently of auction learning.} Both directions are important for future~work.

\vspace{1em}

The remainder of this paper is structured as follows. In  \cref{model}, we formally present the model and setting. In \cref{results}, we formally state our results, which are informally stated above as \cref{intro-bic,intro-single-bidder,intro-single-parameter}. In \cref{overview}, we overview the main ideas behind the proof of \cref{intro-bic}, a proof that we give in full detail in \cref{bic}.
In \cref{exact}, we derive \cref{intro-single-bidder,intro-single-parameter}. We present some extensions in \cref{extensions}.
In \cref{dsic}, we state and prove a result analogous to \cref{intro-bic} for DSIC auctions, using similar proof techniques. Parts of certain proofs are relegated to \cref{proofs}.

\section{Model and Preliminaries}\label{model}

\paragraph{The Decision Maker (Seller), Bidders, and Outcomes.} A single decision maker has the power to choose a social outcome, such as who gets which good that is for sale, or such as which pastime activities are offered in which of the weekends of the upcoming year. There are $n$ bidders who have stakes in this outcome. (The decision maker will be able to charge the bidders and will wish to maximize her revenue.) The possible set of allowed outcomes is denoted by $X$ and can be completely arbitrary. A central example is that of an $m$-parameter auction: the decision maker is a seller who has $m$ items for sale, and the set of outcomes/allocations is $X\subseteq[0,1]^{n\cdot m}$, where an allocation $\vec{x}=(x_{i,j})_{i\in[n],j\in[m]}\in X$ specifies for each bidder $i$ and good~$j$ the amount of good $j$ that bidder $i$ wins. The traditional multi-item setting is the special case with $X={X_{\text{multi-item}}\eqdef\bigl\{(x_{i,j})_{i,j}\in\{0,1\}^{n\cdot m}\mid\forall j:\sum_{i=1}^n x_{i,j}\le1\bigr\}}$, while outcomes with fractional coordinates occur for example in the canonical model of position auctions, where smaller coordinates denote smaller click-through rates.

\begin{sloppypar}
\paragraph{Values.} Bidder $i\in[n]$ has a valuation function $v_i(\cdot)$ over the set of possible outcomes $X$. This function is parametrized by $m$ values $v_{i,1},\ldots, v_{i,m}$ (we will not explicitly write $v_{i,\vec{v}_i}(\cdot)$, but refer to the parameters implicitly for ease of notation. Moreover, as $v_i(\cdot)$ is completely determined by $v_{i,1},\ldots, v_{i,m}$, we will sometimes simply refer to $\vec{v}_i$ as bidder $i$'s value, and to~$v_{i,j}$ as bidder $i$'s value for parameter~$j$) and drawn from a given distribution such that:
\begin{itemize}
\item (Independent Parameters)
The $v_{i,j}$s are independent random variables, drawn from distributions $V_{i,j}$ which are all supported in $[0,H]$.
\item (Lipschitz) There exists an absolute constant $L$, such that if $v'_{i,1},\ldots,v'_{i,m}$ is obtained from $v_{i,1},\ldots,v_{i,m}$ by modifying one of the $v_{i,j}$s by at most an additive $\varepsilon$, then ${|v_i(x) - v'_i(x)|}\leq
L\varepsilon$ for all $x\in X$.
\end{itemize}
For example, in the multi-item setting described above, $v_i(x)=\sum_{j=1}^m x_{i,j}\cdot v_{i,j}$ (and $L=1$).\footnote{This case is actually Lipschitz in a stronger sense: if $v'_{i,1},\ldots,v'_{i,m}$ is obtained from $v_{i,1},\ldots,v_{i,m}$ by modifying one of the $v_{i,j}$s by at most an additive $\varepsilon$, then $|v_i(\vec{x}) - v'_i(\vec{x})|\leq|\vec{x}|_1\cdot L\varepsilon$ for $L=1$. We note that using this stronger property (as well as other properties of the multi-item setting such as monotonicity), our analysis (mutatis mutandis) can be used to quantitatively improve the polynomial dependency of our sample complexity on the parameters of the problems, however we do not follow this direction in this paper. In general, in this paper we always choose generality of results over tighter polynomials.} An additive-up-to-$k$-items setting may also be easily captured using this setting (again with $L=1$), and so can even settings with complementarities, such as a setting in which good~$2j$ is worth $v_{i,j}$ to bidder $i$ iff bidder $i$ also gets good $2j+1$ (and is otherwise worthless to bidder $i$). Since $X$ can be completely arbitrary (in particular does not have to be a subset of~$[0,1]^{n\cdot m}$), we can most generally think of the $v_{i,j}$s as parameters that capture the ``relevant attributes'' of each bidder, such as affinity to action films, affinity to winter sports, willingness to spend a lot of time in a pastime activity, etc. The only requirement is that these attributes are a ``natural'' parametrization in the sense that the utility of the bidder from any given outcome $x\in X$ smoothly depends on (i.e., is Lipschitz in) each of them,\footnote{This rules out such ``tricks'' as using bit-interleaving to condense the $m$ parameters into a single parameter.} and that they are independently drawn.
\end{sloppypar}

We note that both properties above (independent items and Lipschitz) together imply that the valuation of each bidder for each outcome is bounded in $[0,mLH]$.

\paragraph{Payments, Priced Outcomes, and Mechanisms.} A payment specification $p=(p_i)_{i\in[n]}$ specifies for each bidder $i$ to be charged $p_i$. A priced outcome is a pair $(x,p)$ of an allocation and a payment specification. The utility of bidder $i$ with value $v_i(\cdot)$ from priced outcome $(x,p)$ is $u_i\bigl(v_i,(x,p)\bigr)=v_i(x) - p_i$.
An auction/mechanism is a function that maps each valuation profile $(v_{i,j})_{i \in [n], j \in [m]}$ to a distribution over priced outcomes.
The seller's expected revenue from a mechanism $\mu$ is\footnote{This expectation is both over the draw of valuation profile $v$ and over the draw of priced outcome from $\mu(v)$; to avoid clutter, we will not explicitly mention the latter in our notation throughout this paper.} $\expect{v\sim\bigtimes_{i,j}V_{i,j}}{\sum_{i\in[n]}p_i(v)}$, where $p(v)$ is the payment specification chosen by the mechanism for the valuation profile $v$.

\paragraph{Truthfulness.} An auction $\mu$ is individually rational (IR) if the expected utility of a truthful bidder is nonnegative at any valuation profile, i.e.:
$\expect{}{u_k\bigl(v_k,\mu(v)\bigr)}\ge0 \text{~for every~} k\in[n] \text{~and~} v\in[0,H]^{n\cdot m},$
where the expectation is over the randomness of the auction.
For $\varepsilon>0$, an auction $\mu$ is $\varepsilon$-dominant-strategy incentive compatible ($\varepsilon$-DSIC) if truthful bidding maximizes a bidder's expected utility at any valuation profile up to an additive $\varepsilon$, i.e.:\footnote{As is standard in the Game Theory literature, given an $n$-bidder valuation profile $v=(v_1,\ldots,v_n)$, we use $v_{-k}$ for any $k\in[n]$ to denote this valuation profile without the valuation of bidder $k$, and use $(v'_k,v_{-k})$ for any $v'_k\in[0,H]^m$ to denote the valuation profile obtained from $v$ by replacing bidder $k$'s valuation with $v'_k$.}
$\expect{}{u_k\bigl(v_k,\mu(v)\bigr)}\ge\expect{}{u_k\bigl(v_k,\mu(v'_k,v_{-k})\bigr)}-\varepsilon$ for every $k\in[n]$, $v\in[0,H]^{n\cdot m}$, and $v'_k\in[0,H]^m$,
where the expectation is once again over the randomness of the auction.
An auction is DSIC if it is $0$-DSIC.
An auction $\mu$ is $\varepsilon$-Bayesian incentive compatible ($\varepsilon$-BIC) if truthful bidding maximizes a bidder's utility in expectation over all valuations of the other bidders, up to an additive $\varepsilon$, i.e.:
$\expect{v_{-k}\sim\bigtimes_{\substack{i,j\\i\ne k}}V_{i,j}}{u_k\bigl(v_k,\mu(v_k,v_{-k})\bigr)}\ge\expect{v_{-k}\sim\bigtimes_{\substack{i,j\\i\ne k}}V_{i,j}}{u_k\bigl(v_k,\mu(v'_k,v_{-k})\bigr)}-\varepsilon$ for every $k\in[n]$ and $v_k,v'_k\in[0,H]^m$,
where the expectation is both over the valuations of the bidders other than $k$ and over the randomness of the auction.
An auction is BIC if it is $0$-BIC.

\paragraph{Additional Notation.} We will use the following additional notation in our analysis, where $\varepsilon>0$:
\begin{itemize}
\item
For $v\in[0,H]$, we denote by $\epsfloor{v}$ the value of $v$, rounded down to the nearest integer multiple of $\varepsilon$.
\item
We use $[0,H]_{\varepsilon}=\bigl\{\epsfloor{v}\mid v\in [0,H]\bigr\}$ to denote the set of integer multiples of $\varepsilon$ in $[0,H]$.
\item
For every $i,j$, we denote by $\epsfloor{V_{i,j}}$ the distribution of $\epsfloor{v_{i,j}}$ for $v_{i,j}\sim V_{i,j}$.
\end{itemize}

\paragraph{Existing Tools.} In our analysis, we will make use of the following two theorems, which we state below in a way that is adapted to the notation of our paper. The first shows the optimal revenue over all $\varepsilon$-BIC auctions and the optimal revenue over all BIC auctions are close (while this is stated in \citet{rw15} with respect to multi-parameter settings with allocations in $\{0,1\}^{n\cdot m}$, the same proof holds verbatim for arbitrary outcome sets $X$):

\begin{theorem}[\citealp{rw15};\footnote{If we denote by $R_{\varepsilon}$ the maximum expected revenue attainable by any IR and $\varepsilon$-BIC auction for the bidders' product distribution, and by $R$ the maximum revenue attainable by any IR and BIC auction for the same distribution, then the result of \citet{rw15} is that for any $\eta>0$, it is the case that $R\ge(1-\eta)\cdot(R_{\varepsilon}-\frac{n\varepsilon}{\eta})$. Choosing $\eta=\sqrt{\frac{\varepsilon}{mLH}}$ yields \cref{eps-bic-bic} as stated above, since $R_{\varepsilon}$ is trivially bounded from above by the maximum possible sum of valuations, i.e., by $nmLH$.} see also \citealp{dw12}]\label{eps-bic-bic}Let $\mathcal{D}$ be any joint distribution over arbitrary valuations, where the valuations of different bidders are independent. The maximum revenue attainable by any IR and $\varepsilon$-BIC auction for~$\mathcal{D}$ is at most $2n\sqrt{mLH\varepsilon}$ greater than the maximum revenue attainable by any IR and BIC auction for that distribution.
\end{theorem}

\noindent The second is a Chernoff-style concentration inequality for product distributions:

\begin{theorem}[\citealp{bbp17}; see also \citealp{dhp16}]\label{concentration}
Let $D_1,\ldots,D_\ell$ be discrete distributions. Let $S\in\mathbb{N}$. For every $i$, draw $S$ independent samples from $D_i$, and let $D_i^{(S)}$ be the uniform distribution over these samples.
For every $\varepsilon>0$ and $f:\bigtimes_{i=1}^\ell\supp D_i\rightarrow[0,H]$, we have that $\Pr\Bigl(\bigl|\expect{\bigtimes_{i=1}^\ell D_i^{(S)}}{f}-\expect{\bigtimes_{i=1}^\ell D_i}{f}\bigr|>\varepsilon\Bigl)\le\frac{4H}{\varepsilon}\exp\bigl(-\frac{\varepsilon^2S}{8H^2}\bigr).$
\end{theorem}

Unlike standard Chernoff bounds, which state that the expectation over the empirically sampled distribution (over $\ell$-tuples) well approximates the expectation over the true distribution, this  concentration inequality states that the expectation over the \emph{product of the marginals of the empirically sampled distribution} well approximates the expectation over the true distribution. This difference is crucial to us since, as we will see, there are far less such possible products of empirical marginals than there are empirical distributions, so a far smaller number of auctions can guarantee revenue optimization over the former, and this concentration inequality lets us generalize this revenue optimization to be over the true distribution.

\section{Main Results}\label{results}

In this \lcnamecref{results}, we formally state our main results, which were informally presented as \cref{intro-bic,intro-single-bidder,intro-single-parameter} in the introduction. We start with our main result.

\begin{sloppypar}
\begin{theorem}[Main Result]\label{polynomial-bic}
For every $\varepsilon,\delta>0$ and for every $\eta\le\poly(n,m,L,H,\varepsilon)$, the sample complexity of learning an up-to-$\varepsilon$ optimal IR and $\eta$-BIC auction is $\poly(n,m,L,H,\nicefrac{1}{\varepsilon},\nicefrac{1}{\eta},\log\nicefrac{1}{\delta})$. That is, there exists a deterministic algorithm\footnote{\label{oracle}Recall that this result is information-theoretic and not computationally efficient (by necessity, without resolving major open problems), so our decision maker (seller) is computationally unbounded, and we allow the algorithm to make calls to any deterministic oracle that has no access to any $V_{i,j}$. In particular, we assume access to an oracle that can solve the revenue maximization problem on any precisely given $V'_{i,j}$ of finite support.} that given $\poly(n,m,L,H,\nicefrac{1}{\varepsilon},\nicefrac{1}{\eta},\log\nicefrac{1}{\delta})$ samples from each $V_{i,j}$, with probability $1\!-\!\delta$ outputs an IR and $\eta$-BIC auction that attains from $\bigtimes_{i,j} V_{i,j}$ expected revenue at most an additive~$\varepsilon$ smaller than any IR and $\eta$-BIC auction.
\end{theorem}
\end{sloppypar}

The following corollary of our main result should be contrasted with a result of \citet{dhn14}, which shows that finding an $\varepsilon$-optimal mechanism for a single additive bidder with \emph{correlated} item distributions requires exponentially many samples.

\begin{theorem}[Single Bidder]\label{polynomial-single-bidder}
When there is $n=1$ bidder, for every $\varepsilon,\delta>0$, the sample complexity of learning an up-to-$\varepsilon$ optimal IR and IC\footnote{Recall that for a single bidder, the notions of BIC and DSIC coincide.} auction is $\poly(n,m,L,H,\nicefrac{1}{\varepsilon},\log\nicefrac{1}{\delta})$. That is, there exists a deterministic algorithm\footnote{See \cref{oracle}.}
that given $\poly(n,m,L,H,\nicefrac{1}{\varepsilon},\log\nicefrac{1}{\delta})$ samples from each $V_{i,j}$, with probability $1\!-\!\delta$ outputs an IR and IC auction that attains from $\bigtimes_{i,j} V_{i,j}$ expected revenue at most an additive~$\varepsilon$ smaller than any IR and IC auction.
\end{theorem}

The following corollary of our main result unifies and even somewhat extends the state-of-the-art results for single-parameter ($m=1$) revenue maximization. To state it we restrict ourselves to the setting where revenue maximization has been solved by \citet{m81}: assume that $X\subseteq[0,1]^n$ and that $v_i(x)=v_i\cdot x_i$.\footnote{This setting is slightly more general than the state-of-the-art single-parameter results, which assume $X\subseteq\{0,1\}^n$ and/or the special setting of position auctions \citep{ht16,gn17}. As noted in the introduction, \citet{gn17} explicitly state that their techniques cannot accommodate arbitrary $X\subseteq[0,1]^n$ and leave this question (which we successfully resolve) open.}

\begin{sloppypar}
\begin{theorem}[Single-Parameter]\label{polynomial-single-parameter}
In an \mbox{$m\!=\!1$}-parameter setting with $X\subseteq[0,1]^n$ and ${v_i(x)=v_i\cdot x_i}$, for every ${\varepsilon,\delta>0}$, the sample complexity of \textbf{efficiently} learning an \mbox{up-to-$\varepsilon$} IR and DSIC auction is $\poly(n,m,L,H,\nicefrac{1}{\varepsilon},\log\nicefrac{1}{\delta})$. That is, there exists a deterministic algorithm with running time $\poly(n,m,L,H,\nicefrac{1}{\varepsilon},\log\nicefrac{1}{\delta})$ that given $\poly(n,m,L,H,\nicefrac{1}{\varepsilon},\log\nicefrac{1}{\delta})$ samples from each $V_{i,j}$, with probability $1\!-\!\delta$ outputs an IR and DSIC auction that attains from $\bigtimes_{i,j} V_{i,j}$ expected revenue at most an additive~$\varepsilon$ smaller than any IR and BIC auction.
\end{theorem}
\end{sloppypar}

\section{Main Result Proof Overview}\label{overview}

In this \lcnamecref{overview} we roughly sketch our learning algorithm and present each of the main ideas behind its analysis, by presenting a proof overview structured to present each of these ideas separately. The proof overview is given in this \lcnamecref{overview} only for an additive multi-item setting, and some elements of the proof are omitted or glossed over for readability. The full proof, which contains all omitted details and applies to a general arbitrary Lipschitz setting, and in which the main ideas that are surveyed in this \lcnamecref{overview} separately are quite intermingled, is given in \cref{bic}.\footnote{The somewhat less involved proof for DSIC auctions that is sketched in this \lcnamecref{overview} as an intermediary proof is given in \cref{dsic} (in that \lcnamecref{dsic}, though, the analysis is more general as it refers to general Lipschitz valuations and not only to additive valuations as in this \lcnamecref{overview}), both since we find that result interesting in its own right, and to allow interested readers to familiarize themselves with that proof before diving into the more involved proof for BIC auctions in \cref{bic}.}

Our learning algorithm is similar in nature to the one presented in \citet{dhp16} for certain single-parameter environments, however the analysis that we will use to show that it does not overfit the samples is completely different (even for single-parameter environments, where our analysis holds for arbitrary allocation constraints).
Recall that our result is (necessarily) information-theoretic and not computationally efficient. Therefore, some steps in the algorithm perform operations that are not known to be performable in poly-time (but can certainly be performed without access to any $V_{i,j}$). In particular, our algorithm will solve an instance of a Bayesian revenue maximization problem for a precisely given input of finite support (step 2). 

\paragraph{Algorithm.} We start with $S$ (to be determined later) independent samples from each $V_{i,j}$. Our algorithm roughly proceeds as follows:
\begin{enumerate}
\item
For each bidder $i$ and good $j$, round all samples from $V_{i,j}$ down to the nearest multiple of $\varepsilon$. Denote the uniform distribution over these rounded samples by $W_{i,j}$.
\item
Find an IR and $O(\varepsilon)$-IC (see below) multi-item auction that maximizes the revenue from the product of the rounded empirical distributions $W_{i,j}$. Denote this auction by~$\mu$. \item Return the auction $\mu^{\varepsilon}$, which on input $\vec{v}$, rounds down all actual bids to the nearest multiple of $\varepsilon$, $\lfloor\vec{v}\rfloor_\varepsilon$, and allocates and charges payments according to the output of $\mu(\lfloor\vec{v}\rfloor_\varepsilon)$ when run on these rounded bids.
\end{enumerate}

We start with a simpler scenario, namely that of learning DSIC auctions, and only toward the end of this \lcnamecref{overview} introduce the additional issues that arise when learning BIC auctions. We therefore start by showing that if in step~2 of our algorithm we interpret ``IC'' as ``DSIC,'' that is, that if in that step we take an IR and $O(\varepsilon)$-DSIC auction that maximizes the revenue from the product of the rounded empirical distributions, then there exists $S=\poly(n,m,H,\nicefrac{1}{\varepsilon},\log\nicefrac{1}{\delta})$ such that the auction~$\mu^{\varepsilon}$ output by our algorithm is $O(\varepsilon)$-DSIC and its revenue from $\bigtimes_{i,j}V_{i,j}$ is, with probability at least $1\!-\!\delta$, up-to-$O(\varepsilon)$-close to the maximum revenue attainable from $\bigtimes_{i,j}V_{i,j}$ by any DSIC auction. (The formal statement and full proof are given in \cref{dsic}.) We note that the auction output by the algorithm is indeed $O(\varepsilon)$-DSIC, since the output $\mu$ in step~2 is $O(\varepsilon)$-DSIC, and the rounding of the actual bids as defined in step~3 only loses another~$m\varepsilon$.\footnote{\label{big-O-gloss}The astute reader will notice that $m\varepsilon \notin O(\varepsilon)$. As all our bounds are polynomial in $m, \nicefrac{1}{\varepsilon}$ anyway, this is immaterial, and one example of a detail that we glossed over in this \lcnamecref{overview} in the interest of cleanliness, as promised.}

\paragraph{Uniform Convergence of the Revenue of all Possible Output Mechanisms.} Note that for every $i,j$, each rounded sample from step~1 of the algorithm is independently distributed according to $\epsfloor{V_{i,j}}$. The main challenge is in showing that the resulting auction gives up-to-$O(\varepsilon)$ optimal revenue not only on the rounded empirical distributions $\bigtimes_{i,j}W_{i,j}$, but also on the rounded true distributions $\bigtimes_{i,j}\epsfloor{V_{i,j}}$. That is, the main challenge is in showing that no overfitting occurs, in the absence of any structural properties that we can exploit for the mechanisms that are optimal (or up-to-$O(\varepsilon)$ optimal) for $\bigtimes_{i,j}W_{i,j}$. 

This is the point where our approach makes a sharp departure from prior works. Prior work deems this task to be hopeless, and proceeds by proving structural results on optimal mechanisms for restricted domains. We circumvent this by instead simply counting the number possible inputs we will ever query in step~2, and observing that the number of mechanisms over which we have to obtain uniform convergence is at most this number. A crucial observation is that while we do have to consider more and more mechanisms as the number of samples~$S$ grows, the number of mechanisms that we have to consider grows moderately enough so as to not eclipse our gains from increasing the number of samples that we take. For this argument to hold, it is \emph{essential} that our distributions are product distributions.

Let $\mathcal{V}$ be the set of all product distributions $\bigtimes_{i,j} W'_{i,j}$ where each $W'_{i,j}$ is the uniform distribution over some multiset of $S$ values from $[0,H]_{\varepsilon}$. Let $M$ be the set of all mechanisms returned by step~2 of the algorithm for some distribution $\bigtimes_{i,j} W'_{i,j} \in \mathcal{V}$.
\textbf{\boldmath At the heart of our analysis, and of this part of our analysis in particular, is the observation that $|\mathcal{V}|\leq(S+1)^{n\cdot m\cdot\lceil\nicefrac{H}{\varepsilon}\rceil}$. Crucially, this expression has $S$ only in the base and not in the exponent.} Indeed, for every $\bigtimes_{i,j} W'_{i,j} \in \mathcal{V}$, for every $i,j$, and for every integer multiple of $\varepsilon$ in $[0,H]$ (there are $\lceil\nicefrac{H}{\varepsilon}\rceil$ many such values), the probability of this value in $W'_{i,j}$ can be any of the $S+1$ values $0,\nicefrac{1}{S},\ldots,1$. Therefore, $|M|\leq(S+1)^{n\cdot m\cdot\lceil\nicefrac{H}{\varepsilon}\rceil}$. 

We will choose $S$ so that with probability at least $1\!-\!\delta$, it simultaneously holds for all mechanisms $\mu\in M$ that
\begin{equation}\label{concentrate}
|\Rev_{\bigtimes_{i,j}W_{i,j}}(\mu)-\Rev_{\bigtimes_{i,j}\epsfloor{V_{i,j}}}(\mu)|\le \varepsilon.
\end{equation}

To this end, we will use a Chernoff-style concentration bound (\cref{concentration}) for product distributions, which when applied to our setting shows that for each auction separately \cref{concentrate} is violated with probability exponentially small in $\frac{\varepsilon^2S}{m^2H^2}$. So, to have \cref{concentrate} hold with probability at most $1\!-\!\delta$ for all auctions in $M$ simultaneously, we choose $S$ so that the violation probability for each auction separately is at most $\nicefrac{\delta}{|M|}$, and use a union bound.
Since $|M|\leq(S+1)^{n\cdot m\cdot\lceil\nicefrac{H}{\varepsilon}\rceil}$, we have that it is enough to take $S$ such that $\frac{\varepsilon^2S}{m^2H^2}$ is of order of magnitude at least $\log\nicefrac{|M|}{\delta}=\log\frac{(S+1)^{n\cdot m\cdot\lceil\nicefrac{H}{\varepsilon}\rceil}}{\delta}\approx\log\nicefrac{1}{\delta}+n\cdot m\cdot \nicefrac{H}{\varepsilon} \log S$, which is clearly possible by taking a suitable $S$ that is polynomial in $n$, $m$, $H$, $\nicefrac{1}{\varepsilon}$, and $\log\nicefrac{1}{\delta}$.
So, taking a number of sample of this magnitude gives that with probability at least $1\!-\!\delta$, \cref{concentrate} simultaneously holds for all mechanisms in $M$ and so the mechanism output by step~2 of the algorithm gets up-to-$O(\varepsilon)$ the same revenue on the product of the rounded empirical distributions $\bigtimes_{i,j} W_{i,j}$ as it does on the product of the rounded true distributions $\bigtimes_{i,j}\epsfloor{V_{i,j}}$. So, the revenue that the mechanism $\mu^{\varepsilon}$ output by (step~3 of) the algorithm attains from $\bigtimes_{i,j}V_{i,j}$ is identical to the revenue that the mechanism $\mu$ output by step~2 of the algorithm attains from $\bigtimes_{i,j}\epsfloor{V_{i,j}}$, which is up-to-$O(\varepsilon)$ the optimal revenue attainable from $\bigtimes_{i,j} W_{i,j}$.

\paragraph{Revenue Close to Optimal.} Our next task is to show that with high probability the optimal revenue attainable from $\bigtimes_{i,j} W_{i,j}$ by any $O(\varepsilon)$-DSIC auction is up-to-$O(\varepsilon)$ the same as the optimal revenue attainable from $\bigtimes_{i,j}V_{i,j}$ by any DSIC auction, which would imply that the revenue that $\mu^\varepsilon$ attains from $\bigtimes_{i,j}V_{i,j}$ is close to optimal, as required. Let $\OPT$ be the DSIC auction that maximizes the revenue (among such auctions) in expectation over $\bigtimes_{i,j}V_{i,j}$. \textbf{\boldmath At the heart of this part of our analysis is the fact that while our algorithm cannot hope to find $\OPT$, we can nonetheless carefully reason about it in our analysis, as it is nonetheless fixed and well-defined (in particular, it does not depend on the drawn samples).} Let $\OPT_{\varepsilon}$ be the mechanism defined over $\bigtimes_{i,j}\epsfloor{V_{i,j}}$ as follows: for each bidder $i$ and item $j$, let $w_{i,j}$ be the input bid of bidder $i$ for item $j$ (a multiple of~$\varepsilon$), and replace it by a bid $v_{i,j}$ independently drawn from the distribution~$V_{i,j}$ conditioned upon being in the interval $[w_{i,j},w_{i,j}+\varepsilon)$; the auction $\OPT_{\varepsilon}$ allocates and charges payments according to the output of $\OPT$  when run on these drawn replacement bids. Obviously, the auction $\OPT_{\varepsilon}$ is an $O(\varepsilon)$-DSIC auction\footnote{Or more accurately, $O(m\varepsilon)$-DSIC; see \cref{big-O-gloss}.} whose revenue from $\bigtimes_{i,j}\epsfloor{V_{i,j}}$ is identical to that of the auction $\OPT$ from $\bigtimes_{i,j}V_{i,j}$, i.e., to the optimal revenue from $\bigtimes_{i,j}V_{i,j}$, so it is enough to show that the revenue of the auction $\OPT_{\varepsilon}$ from $\bigtimes_{i,j}\epsfloor{V_{i,j}}$ and from $\bigtimes_{i,j}W_{i,j}$ is the same up-to-$O(\varepsilon)$ with high probability, that is, that \cref{concentrate} also holds for the mechanism $\OPT_{\varepsilon}$ with high probability. To do so, we modify the definition of the set $M$ to also include the (well-defined even prior to sampling, despite being unknown to our algorithm) mechanism $\OPT_{\varepsilon}$ --- since the order of magnitude of $|M|$ does not change, the order of magnitude of the number of samples required to guarantee that \cref{concentrate} holds for all auctions in $M$ (including $\OPT_{\varepsilon}$) does not change.

\paragraph{Bayesian Incentive Compatibility.} We conclude our proof overview 
by adapting the proof to the more delicate BIC notion of incentive compatibility, thus showing that if in step~2 of our algorithm we take an $O(\varepsilon)$-BIC (rather than $O(\varepsilon)$-DSIC) and IR auction that maximizes the revenue from the product of the rounded empirical distributions, then there exists $S=\poly(n,m,H,\nicefrac{1}{\varepsilon},\log\nicefrac{1}{\delta})$ such that the auction $\mu^{\varepsilon}$ output by our algorithm is, with probability at least $1-\delta$, an $O(\varepsilon)$-BIC auction whose revenue from $\bigtimes_{i,j}V_{i,j}$ is up-to-$O(\varepsilon)$-close to the maximum revenue attainable from $\bigtimes_{i,j}V_{i,j}$ by any BIC auction (and therefore, by \cref{eps-bic-bic}, up-to-$O(\sqrt{\varepsilon})$-close to the revenue attainable from this distribution by any $O(\varepsilon)$-BIC auction)\footnote{\cref{eps-bic-bic} in fact allows us to also reduce to a cleaner oracle, which finds an optimal BIC auction rather than an optimal $O(\varepsilon)$-BIC auction, in \cref{bic}.}.
The challenge here is that (approximate) BIC is a distribution-dependent property of a mechanism (as opposed to DSIC, which is a distribution-agnostic incentive compatibility notion). Indeed, examining our analysis above with \mbox{($\varepsilon$-)DSIC} replaced by \mbox{($\varepsilon$-)BIC}, we note that the resulting analysis falls short of carrying through in two points: it is unclear why $\OPT_{\varepsilon}$ is $O(\varepsilon)$-BIC not only with respect to $\bigtimes_{i,j}V_{i,j}$ but also with respect to $\bigtimes_{i,j} W_{i,j}$, and it is unclear why any mechanism that can be output by step~2 of our algorithm is $O(\varepsilon)$-BIC not only with respect to $\bigtimes_{i,j} W_{i,j}$ but only with respect to $\bigtimes_{i,j}V_{i,j}$.
\textbf{At the heart of this part of our analysis is the observation that the set of all interim expected utilities, of all bidders' possible types, from all possible reported types, in all mechanisms\footnote{The set of all such interim expected utilities for a single mechanism is sometimes referred to as the \emph{reduced form} of the mechanism.} on the one hand is comprised of a small-enough number of random variables to still enable uniform convergence, and on the other hand contains sufficient information to show that incentive constraints do not deteriorate much.} Concretely, we will choose $S$ so that with probability at least~$1-\delta$, simultaneously for all mechanisms in $M$ (including $\OPT_{\varepsilon}$) not only does \cref{concentrate} hold, but also the following holds for every bidder $k\in[n]$ and values $v_k,v'_k\in[0,H]^m_{\varepsilon}$:
\begin{equation}\label{ic-concentrate}
|\expect{v_{-k}\sim\bigtimes_{\substack{i,j\\i\ne k}}W_{i,j}}{u_k\bigl(v_k,\mu(v'_k,v_{-k})\bigr)}-\expect{v_{-k}\sim\bigtimes_{\substack{i,j\\i\ne k}}\epsfloor{V_{i,j}}}{u_k\bigl(v_k,\mu(v'_k,v_{-k})\bigr)}|\le \varepsilon.
\end{equation}
We note that for every mechanism $\mu$, we require that \cref{ic-concentrate} hold for $n\!\cdot\!\lceil\nicefrac{H}{\varepsilon}\rceil^{2m}$ distinct combinations of of $k\in[n]$ and $v_k,v'_k\in[0,H]_{\varepsilon}^m$. Crucially, this number does not depend on~$S$. So, the number of instances of \cref{ic-concentrate} that we would like to hold simultaneously with high probability is $|M|\cdot n\!\cdot\!\lceil\nicefrac{H}{\varepsilon}\rceil^{2m}$,
and so we have $|M|\cdot(1+n\!\cdot\!\lceil\nicefrac{H}{\varepsilon}\rceil^{2m})\le(1+n\!\cdot\!\lceil\nicefrac{H}{\varepsilon}\rceil^{2m})\cdot(S+1)^{n\cdot m\cdot\lceil\nicefrac{H}{\varepsilon}\rceil}$ instances of either \cref{concentrate} or \cref{ic-concentrate} that we would like to hold simultaneously with high probability.\footnote{A somewhat similar idea appeared in~\citet{cdw12}, albeit without exploiting independence across items.} As this number still has $S$ only in the base and not in the exponent, we can proceed as above to guarantee this with high probability using only a polynomial number of samples.

\section{Proof of Main Result}\label{bic}

In this \lcnamecref{bic}, we give the full details of the proof of our main result, \cref{polynomial-bic}. The proofs of supporting \lcnamecrefs{bic-opteps-opt} are relegated to \cref{proofs}.

\begin{proof}[Proof of \cref{polynomial-bic}]
We assume that for every $i\in[n]$ and $j\in[m]$, we have~$S$ (to be determined later) independent samples $(v_{i,j}^s)_{s=1}^S$ from $V_{i,j}$. \cref{bic-alg} presents our learning \lcnamecref{bic-alg}, which is similar in nature to the one presented in \citet{dhp16} for certain single-parameter environments, however the analysis that we will use to show that it does not overfit the samples is completely different (even for single-parameter environments, where our analysis holds for arbitrary allocation constraints).
\begin{algorithm}[t]
\small
\DontPrintSemicolon
\SetAlgoHangIndent{3.65em}
\SetKwProg{Fn}{Function}{:}{}
\SetKwFunction{EmpOptBIC}{EmpiricalOptimize}
\SetKwFunction{OptOrcBIC}{OptimizationOracle}

\Fn{\EmpOptBIC{$H,X,\varepsilon,\delta,(v^s_{i,j})_{i\in[n],j\in[m]}^{s\in[S]}$}}{
\KwIn{For every $i\in[n],j\in[m]$, $(v^s_{i,j})_{s=1}$ is a sequence of $S=\tilde{O}\left(\tfrac{H^2}{\varepsilon^2}\cdot\bigl(\log\nicefrac{1}{\delta}+\tfrac{nmH}{\varepsilon}\bigr)\right)$ samples from $V_{i,j}$}
\KwOut{With probability $1\!-\!\delta$, an IR and $(4mL\varepsilon)$-BIC mechanism for $\bigtimes_{i,j}V_{i,j}$, defined over $[0,H]^{n\cdot m}$ with allocations in $X$, whose expected revenue from $\bigtimes_{i,j}V_{i,j}$ is up to an additive $\bigl(4nmL(\varepsilon+\sqrt{2H\varepsilon})\bigr)$ smaller than that of any IR and $(4mL\varepsilon)$-BIC mechanism for $\bigtimes_{i,j}V_{i,j}$ with allocations in $X$}
\For{$i\in[n],j\in[m]$}{
\For{$s\in[S]$}{
$w^s_{i,j}\longleftarrow\epsfloor{v^s_{i,j}}$\;
}
$W_{i,j}\longleftarrow$ \emph{the uniform distribution over $(w^s_{i,j})_{s=1}^S$}\;
}
$\mu\longleftarrow\OptOrcBIC{$H,\varepsilon,(W_{i,j})_{i\in[n],j\in[m]}$}$\tcp*{See definition below}
\KwRet{The mechanism that for input $(v_{i,j})_{i\in[n],j\in[m]}$ outputs $\mu\Bigl(\bigl(\epsfloor{v_{i,j}}\bigr)_{i\in[n],j\in[m]}\Bigr)$, modified to charge each bidder $mL\varepsilon$ less.}
}

\vspace{1em}

\Fn{\OptOrcBIC{$H,X,\varepsilon,(W_{i,j})_{i\in[n],j\in[m]}$}}{
\KwIn{For every $i\in[n],j\in[m]$, $W_{i,j}$ is a distribution over $[0,H]_{\varepsilon}$}
\KwOut{An IR and BIC mechanism for $\bigtimes_{i,j}W_{i,j}$, defined over $[0,H]_{\varepsilon}^{n\cdot m}$ with allocations in $X$, which maximizes the expected revenue from $\bigtimes_{i,j}W_{i,j}$ among all IR and BIC mechanisms for $\bigtimes_{i,j}W_{i,j}$ with allocations in $X$}
$\mu\longleftarrow$ \emph{an IR and BIC mechanism for $\bigtimes_{i,j}W_{i,j}$ (defined over $\bigtimes_{i,j}\supp W_{i,j}$) with \linebreak \vphantom{~}\hspace{-.8em}allocations in $X$, which maximizes the expected revenue from $\bigtimes_{i,j}W_{i,j}$ among \linebreak \vphantom{~}\hspace{-.8em}all such mechanisms}\;
\KwRet{The mechanism obtained by extending $\mu$ to be defined over $(v_{i,j})_{i,j}\in [0,H]_{\varepsilon}^{n\cdot m}$ as follows: for every $k\in[n]$ s.t.\ $v_{k,j}\notin\supp W_{k,j}$ for some $j\in[m]$, replace the entire bid vector $v_k=(v_{k,j})_{j\in[m]}$ of bidder $k$ with a bid vector $v'_k=(v'_{k,j})_{j\in[m]}\in\supp\bigtimes_jW_{k,j}$ that maximizes $\expect{v'_{-k}\sim\bigtimes_{i,j:i\ne k}W_{i,j}}{u_k\bigl(v_k,\mu(v'_k,v'_{-k})\bigr)}$}
}

\caption{Empirical Multi-Parameter Up-to-$\varepsilon$ BIC Revenue Maximization.}\label{bic-alg}
\end{algorithm}%
We now analyze \cref{bic-alg}. Note that for every $i,j,s$, we have that $w_{i,j}^s\sim\epsfloor{V_{i,j}}$ independently.

Let $\mathcal{V}$ be the set of all product distributions $\bigtimes_{i,j} W'_{i,j}$ where each $W'_{i,j}$ is the uniform distribution over some multiset of $S$ values from $[0,H]_{\varepsilon}$. Let $M$ be the set of all mechanisms of the form \OptOrcBIC{$H,X,\varepsilon,\bigtimes_{i,j} W'_{i,j}$} for all $\bigtimes_{i,j} W'_{i,j} \in \mathcal{V}$.
At the heart of our analysis is the observation that $|\mathcal{V}|<(S+1)^{n\cdot m\cdot\lceil\nicefrac{H}{\varepsilon}\rceil}$. (Crucially, this expression has~$S$ only in the base and not in the exponent!) Indeed, for every $\bigtimes_{i,j} W'_{i,j} \in \mathcal{V}$, for every~$i,j$, and for every integer multiple of $\varepsilon$ in $[0,H]$ (there are $\lceil\nicefrac{H}{\varepsilon}\rceil$ many such values), the probability of this value in $W'_{i,j}$ can be any of the $S+1$ values $0,\nicefrac{1}{S},\ldots,1$. (The inequality on $|\mathcal{V}|$ is strict since, for example, not all probabilities can be $0$ simultaneously.)
Therefore, $|M|<(S+1)^{n\cdot m\cdot\lceil\nicefrac{H}{\varepsilon}\rceil}$. 

Let $\OPT$ be the IR and $(4mL\varepsilon)$-BIC auction for $\bigtimes_{i,j}V_{i,j}$ that maximizes the expected revenue (among such auctions) from $\bigtimes_{i,j}V_{i,j}$ (our learning \lcnamecref{bic-alg} cannot hope to find $\OPT$, but in our analysis we may carefully reason about it, as it is nonetheless well defined; in particular, it does not depend on $S$). Let $\OPT_{\varepsilon}$ be the (randomized) mechanism defined over $\bigtimes_{i,j}\epsfloor{V_{i,j}}$ as follows: let $w=(w_{i,j})_{i\in[n],j\in[m]}$ be an input valuation; for each $i,j$, independently draw $v_{i,j}\sim V_{i,j}|_{[w_{i,j},w_{i,j}+\varepsilon)}$; let $(x,p)=\OPT\bigl((v_{i,j})_{i\in[n],j\in[m]}\bigr)$; the allocation of $\OPT_{\varepsilon}(w)$ is $x$, and the payment of each bidder $i$ is
$p_i-mL\varepsilon$.

\begin{lemma}\label{bic-opteps-opt}
$\OPT_{\varepsilon}$ is an IR and $(6mL\varepsilon)$-BIC mechanism for $\bigtimes_{i,j}\epsfloor{V_{i,j}}$, whose expected revenue from $\bigtimes_{i,j}\epsfloor{V_{i,j}}$ is
$nmL\varepsilon$ smaller than the expected revenue of $\OPT$ from $\bigtimes_{i,j}V_{i,j}$.
\end{lemma}

We will choose $S$ so that with probability at least $1\!-\!\delta$, both of the following simultaneously hold for all mechanisms $\mu\in M \cup \{\OPT_{\varepsilon}\}$:
\begin{itemize}
\item
\hfill
$\bigl|\Rev_{\bigtimes_{i,j}W_{i,j}}(\mu)-\Rev_{\bigtimes_{i,j}\epsfloor{V_{i,j}}}(\mu)\bigr|\le nmL\varepsilon,$
\hfill\refstepcounter{equation}\label{bic-concentrate}\textup{(\theequation)}
\item
For every agent $k\in[n]$ and types $v_k,v'_k\in[0,H]_{\varepsilon}^m$:
\begin{equation}\label{bic-ic-concentrate}
\bigl|\expect{v_{-k}\sim\bigtimes_{i,j:i\ne k}W_{i,j}}{u_k\bigl(v_k,\mu(v'_k,v_{-k})\bigr)}-\expect{v_{-k}\sim\bigtimes_{i,j:i\ne k}\epsfloor{V_{i,j}}}{u_k\bigl(v_k,\mu(v'_k,v_{-k})\bigr)}\bigr|\le mL\varepsilon.
\end{equation}
\end{itemize}

We note that for every mechanism $\mu$, we require that \cref{bic-ic-concentrate} hold for $n\!\cdot\!\lceil\nicefrac{H}{\varepsilon}\rceil^{2m}$ distinct combinations of of $k\in[n]$ and $v_k,v'_k\in[0,H]_{\varepsilon}^m$. Crucially, this number does not depend on $S$.

By \cref{concentration} (with $\ell\eqdef n\cdot m$, and note that any mechanism's revenue is bounded by
$nmLH$), we have that for each mechanism $\mu\in M \cup \{\OPT_{\varepsilon}\}$ separately \cref{bic-concentrate} holds with probability at least $1-\frac{4He^{-\varepsilon^2S/(8H^2)}}{\varepsilon}$, and for each combination of $(\mu,k,v_k,v'_k)$ separately \cref{bic-ic-concentrate} holds with probability at least $1-\frac{4He^{-\varepsilon^2S/(8H^2)}}{\varepsilon}$.

Choosing $S$ so that each of these probabilities is at least $1-\frac{\delta}{(|M|+1)\cdot n\cdot\lceil\nicefrac{H}{\varepsilon}\rceil^{2m}}$, we obtain that both \cref{bic-concentrate} holds simultaneously for all mechanisms $\mu\in M \cup \{\OPT_{\varepsilon}\}$ and \cref{bic-ic-concentrate} holds simultaneously for all combinations $(\mu,k,v_k,v'_k)$ with probability at least~$1\!-\!\delta$. We now estimate~$S$. Since $|M\cup\{OPT_{\varepsilon}\}|\leq(S+1)^{n\cdot m\cdot\lceil\nicefrac{H}{\varepsilon}\rceil}$, we have that it is enough to take $S$ such that
\[
S\ge\tfrac{8H^2}{\varepsilon^2}\cdot\bigl(\log\tfrac{4H}{\varepsilon}+\log\nicefrac{1}{\delta}+\log n+2m\log\lceil\nicefrac{H}{\varepsilon}\rceil+nm\lceil\nicefrac{H}{\varepsilon}\rceil\log(S+1)\bigr).
\]
Therefore,\footnote{\label{recursion-solution}The requirement is of the form $S\ge a+b\log S$, so the tight solution is of order $a+b\log b$.} there exists an appropriate
\[
S=O\left(\tfrac{H^2}{\varepsilon^2}\log\nicefrac{1}{\delta}+\tfrac{nmH^3}{\varepsilon^3}\log\tfrac{nmH^3}{\varepsilon^3}\right)=\tilde{O}\left(\tfrac{H^2}{\varepsilon^2}\cdot\bigl(\log\nicefrac{1}{\delta}+\tfrac{nmH}{\varepsilon}\bigr)\right).
\]

Let $\mu$ be the output of the call to $\OptOrcBIC$ in \cref{bic-alg}, and let $\mu^{\varepsilon}$ be the final output of the \lcnamecref{bic-alg} (the output of $\EmpOptBIC$).

\begin{lemma}\label{bic-mu-opteps}
If \cref{bic-ic-concentrate} holds for every mechanism in $M \cup \{\OPT_{\varepsilon}\}$ and for every $(k,v_k,v'_k)$, then:
\begin{itemize}
\item
$\OPT_{\varepsilon}$ is a $(8mL\varepsilon)$-BIC mechanism for $\bigtimes_{i,j}W_{i,j}$.
\item
The expected revenue of $\mu$ from $\bigtimes_{i,j}W_{i,j}$ is at most $4nmL\sqrt{2H\varepsilon}$ smaller than the expected revenue of $\OPT_{\varepsilon}$ from $\bigtimes_{i,j}W_{i,j}$.
\item
$\mu$ is a $(2mL\varepsilon)$-BIC mechanism for $\bigtimes_{i,j}\epsfloor{V_{i,j}}$.
\end{itemize}
\end{lemma}

\begin{lemma}\label{bic-mueps-mu}
$\mu^{\varepsilon}$ is an IR mechanism whose expected revenue from $\bigtimes_{i,j}V_{i,j}$ is $nmL\varepsilon$ smaller than the expected revenue of $\mu$ from $\bigtimes_{i,j}\epsfloor{V_{i,j}}$. Furthermore, if $\mu$ is $(2mL\varepsilon)$-BIC for $\bigtimes_{i,j}\epsfloor{V_{i,j}}$, then $\mu^{\varepsilon}$ is $(4mL\varepsilon)$-BIC for $\bigtimes_{i,j}V_{i,j}$.
\end{lemma}

So, with probability at least $1\!-\!\delta$, we have both that $\mu^{\varepsilon}$ is an IR and $(4mL\varepsilon)$-BIC mechanism (by \cref{bic-mu-opteps,bic-mueps-mu}) and that:
\begin{align*}
&\Rev_{\bigtimes_{i,j}V_{i,j}}(\mu^{\varepsilon})=\tag{by \cref{bic-mueps-mu}}\\
=&\Rev_{\bigtimes_{i,j}\epsfloor{V_{i,j}}}(\mu)-nmL\varepsilon\ge\tag{by \cref{bic-concentrate} and since $\mu\in M$}\\
\ge&\Rev_{\bigtimes_{i,j}W_{i,j}}(\mu)-2nmL\varepsilon\tag{by \cref{bic-mu-opteps}}\ge\\
\ge&\Rev_{\bigtimes_{i,j}W_{i,j}}(\OPT_{\varepsilon})-4nmL\sqrt{2H\varepsilon}-2nmL\varepsilon\tag{by \cref{bic-concentrate} for $\OPT_{\varepsilon}$}\ge\\
\ge&\Rev_{\bigtimes_{i,j}\epsfloor{V_{i,j}}}(\OPT_{\varepsilon})-3nmL\varepsilon-4nmL\sqrt{2H\varepsilon}=\tag{by \cref{bic-opteps-opt}}\\
=&\Rev_{\bigtimes_{i,j}V_{i,j}}(\OPT)-4nmL\varepsilon-4nmL\sqrt{2H\varepsilon}=\\
=&\Rev_{\bigtimes_{i,j}V_{i,j}}(\OPT)-4nmL(\varepsilon+\sqrt{2H\varepsilon}).\tag*{\qedhere}
\end{align*}
\end{proof}

\section{From Approximate to Exact Incentive Compatibility}\label{exact}

In this \lcnamecref{exact}, we derive sample complexity results for exact incentive compatibility for the special cases of a single bidder (\cref{polynomial-single-bidder}) or a single good / single-parameter setting (\cref{polynomial-single-parameter}). As mentioned in the introduction, whether this can also be done for more general settings remains an open question.

\subsection{One Bidder}\label{exact-one-bidder}

In this \lcnamecref{exact}, we will prove \cref{polynomial-single-bidder}.
For a single bidder, the following \lcnamecref{one-bidder-nudge}, which to the best of our knowledge first implicitly appeared in \citet{bbhm05}, where it is attributed to Nisan,\footnote{%
It appears there and in following papers \citep{chk07,hn13,dhn14,bgn17,g18} as part of a two-step reduction sometimes called ``nudge and round'' (this is the ``nudge'' part), which reduces the menu size of a single-bidder auction with negligible revenue loss. To the best of our knowledge, the first reference to this argument as a general $\varepsilon$-IC to IC reduction rather than as part of a ``nudge and round'' operation (where it fixes IC issues resulting from rounding) is in \citet{dw12}, who also attribute it to Nisan following \citet{chk07}, who in turn attribute it to Nisan following \citet{bbhm05}. A similar technique appears in \citet{mp17} within the context of correcting for model misspecification. In this context, our rounding of the empirical distributions can be viewed as a deliberately introduced model misspecification that on the one hand novelly serves as a tool against overfitting and on the other hand is carefully controlled so that its adverse effect is limited to at most an $\varepsilon$ loss in incentive compatibility. In the one bidder case, after guaranteeing that no overfitting occurs, we correct for this loss in incentive compatibility using this technique as in the model misspecification literature.} provides an $\varepsilon$-IC to IC reduction with negligible revenue loss.

\begin{theorem}[Nisan, circa 2005]\label{one-bidder-nudge}
Let $\mu$ be an IR and $\varepsilon$-IC\footnote{Recall once again that for a single bidder, the notions of BIC and DSIC coincide.} mechanism for a single bidder. Modifying each possible priced outcome by multiplying the payment in that priced outcome by~$1\!-\!\sqrt{\varepsilon}$ and letting the bidder choose the (modified) priced outcome that maximizes her utility yields an IR and IC mechanism $\mu'$ with expected revenue at least $(1\!-\!\sqrt{\varepsilon})\cdot(\Rev(\mu)-\sqrt{\varepsilon})$.
\end{theorem}

For completeness, we provide a proof of this \lcnamecref{one-bidder-nudge}.
The idea is that discounting more expensive priced outcomes more heavily makes sure that incentives do not drive the bidder toward a much cheaper priced outcome. More concretely, due to the auction being only $\varepsilon$-IC, the utility of a bidder from choosing a cheaper priced outcome can be higher by at most~$\varepsilon$. Since for any priced outcome whose price is cheaper by more than a $\sqrt{\varepsilon}$ compared to the bidder's original priced outcome, the given discount is smaller by more than $\sqrt{\varepsilon}^2=\varepsilon$, this smaller discount more than eliminates any potential utility gain due to choosing the cheaper priced outcome, so such a cheaper priced outcome would not become the most-preferred one.

\begin{proof}
Fix a type $v\in[0,H]^m$ for the bidder. Let $e$ be the priced outcome (a distribution over priced outcomes, i.e., a random variable, if $\mu$ is randomized) according to $\mu$ when the bidder has type $v$. It is enough to show that the bidder pays at least $(1\!-\!\sqrt{\varepsilon})(p_e-\sqrt{\varepsilon})$ in expectation according to $\mu'$ when he has type $v$. (We denote the expected price of, e.g., $e$ by $p_e$.) Let $f'$ be a possible priced outcome of $\mu'$, and let $f$ be the priced outcome of $\mu$ that corresponds to it. We will show that if $p_{f'}<(1\!-\!\sqrt{\varepsilon})(p_e-\sqrt{\varepsilon})$, then the bidder strictly prefers the priced outcome $e'$ of $\mu'$ that corresponds to $e$ over~$f'$ (and so does not choose $f'$ in $\mu'$, completing the proof). Indeed, since in this case $p_f<p_e-\sqrt{\varepsilon}$, we have that
\begin{multline*}
\expect{}{u(v,e')}=
\expect{}{u(v,e)}+\sqrt{\varepsilon}\cdot p_e \ge
\expect{}{u(v,f)}-\varepsilon+\sqrt{\varepsilon}\cdot p_e=
\expect{}{u(v,f')}-\sqrt{\varepsilon}\cdot p_f-\varepsilon+\sqrt{\varepsilon}\cdot p_e=\\=
\expect{}{u(v,f')}-\varepsilon+\sqrt{\varepsilon}\cdot (p_e-p_f)>
\expect{}{u(v,f')}-\varepsilon+\sqrt{\varepsilon}\cdot \sqrt{\varepsilon}=
\expect{}{u(v,f')},
\end{multline*}
as required.
\end{proof}

\noindent
Applying \cref{one-bidder-nudge} to the auction output by\footnote{Or, for a somewhat simpler analysis, to the auction output by \cref{dsic-alg} in \cref{dsic}. For readers interested in precise polynomial dependencies, in this case the allowed revenue deviation in \cref{dsic-concentrate} could also be loosened to $\sqrt{mL\varepsilon}$ to save on samples, since the revenue loss from running the reduction of \cref{one-bidder-nudge} would be far greater anyway.} \cref{bic-alg} yields \cref{polynomial-single-bidder}.

\subsection{Single-Parameter Settings}

In this \lcnamecref{exact}, we will prove \cref{polynomial-single-parameter}.
The algorithm presented in \cref{bic} constitutes a black-box reductions from $\varepsilon$-BIC revenue maximization from samples to BIC revenue maximization from given distributions. As noted in the introduction, the latter are mostly unsolved for more than one good. For a single good, however, the problem of DSIC/BIC revenue maximization was completely resolved by the seminal work of \citet{m81} (who, in particular, showed that the optimal BIC mechanism is DSIC), and the computation complexity of the solution for discrete distributions was shown by \citet{e07} to be polynomial.

\begin{definition}[Myersonian Auction, \citealp{m81}]
An $n$-bidder \emph{Myersonian auction} (for valuations in $[0,H]$) is a tuple $(\phi_i)_{i\in N}$, where for every $i\in[n]$,\ \ $\phi_i:[0,H]\rightarrow\mathbb{R}$ is a nondecreasing function called the \emph{ironed virtual valuation} of bidder $i$.
The chosen outcome is $x\in X$ that maximizes $\sum_{i\in[n]}x_i\cdot\phi_i(v_i)$, with ties broken in a consistent manner. The payment is defined by the payment identity of \citet{m81}, which guarantees that the auction is IR and DSIC.
\end{definition}

\begin{theorem}[\citealp{m81}]
For every product distribution $\bigtimes_{i=1}^nW_i$, there exists a (DSIC) Myersonian auction $(\phi_i)_{i=1}^n$ that attains maximum revenue from $\bigtimes_{i=1}^nW_i$ among all IR and BIC auctions. Moreover, for every $i\in[n]$, the ironed virtual valuation $\phi_i$ depends only on~$W_i$.
\end{theorem}

\begin{theorem}[\citealp{e07}]\label{one-good-efficient}
Let $S\in\mathbb{N}$. There exists an algorithm that runs in time $\poly(S)$, such that given
a discrete distribution $W$ with support of size at most $S$, outputs a nondecreasing function $\phi:\supp W\rightarrow\mathbb{R}$, such that for every product $\bigtimes_{i=1}^nW_i$ of discrete distributions each having support of size at most $S$, the (DSIC) Myersonian auction $(\phi_i)_{i\in N}$ (where $\phi_i$ is the output of the algorithm given $W_i$) maximizes the expected revenue from $\bigtimes_{i=1}^nW_i$ among all IR and BIC auctions.
\end{theorem}

Plugging\footnote{Extending the mechanism $\mu$ returned by the optimizer in $\OptOrcBIC$ to be defined over all of $[0,H]_{\varepsilon}^{n\cdot m}$ as follows, though \citep[following][]{gn17}: for every $i\in[n]$ and $j\in[m]$ s.t.\ $v_{i,j}\notin\supp W_{i,j}$, replace $v_{i,j}$ with $\max\{w_{i,j}\in\supp W_{i,j} \mid w_{i,j}\le v_{i,j}\}$.} \cref{one-good-efficient} into \cref{bic-alg} brings us closer (by making the \lcnamecref{bic-alg} efficient) to proving \cref{polynomial-single-parameter}, however seems to still result in an $\varepsilon$-DSIC (rather than precisely DSIC) auction. Indeed, in the notation of \cref{bic}, while $\mu$ is DSIC, it seems that $\mu^{\varepsilon}$ is only $O(m\varepsilon)$-DSIC. To complete the proof of \cref{polynomial-single-parameter}, we note that the latter is in fact exactly DSIC in this case. Indeed, its allocation rule is monotone, and
it satisfies the payment identity of \citet{m81} for every bidder.

\section{Extensions}\label{extensions}

\subsection{On Computational Efficiency}

As mentioned above, it is currently not known how to efficiently implement the optimization oracle, outputting an up-to-$O(\varepsilon)$-optimal auction, as used in our algorithm. Nonetheless, there has been quite some work on efficiently finding auctions with weaker revenue guarantees with respect to optimal, such as guaranteeing some constant fraction of the optimal revenue. As the structured argument of our analysis provides a black-box reduction from BIC-revenue-maximization given a full distributions description to $\varepsilon$-BIC-revenue-maximization from samples, we can plug in any such algorithm into our analysis (in lieu of the optimization oracle) to obtain a learning algorithm with matching computational complexity and revenue degradation:

\begin{theorem}[``Meta Theorem'': Black-Box Reduction for Efficient Up-to-Constant Guarantees]
If there exists a polynomial-time algorithm for Bayesian revenue-maximization up to a constant factor $C$ given an explicitly specified finite product distribution, then for every $\varepsilon,\delta>0$ and for every $\eta\le\poly(n,m,L,H,\varepsilon)$, there exists a polynomial-time algorithm that given $\poly(n,m,L,H,\nicefrac{1}{\varepsilon},\nicefrac{1}{\eta},\log\nicefrac{1}{\delta})$ samples from each $V_{i,j}$, with probability at least $1\!-\!\delta$ outputs an IR and $\eta$-BIC auction that attains from $\bigtimes_{i,j} V_{i,j}$ expected revenue at most an additive~$\varepsilon$ smaller than a $C$ fraction of that of any IR and $\eta$-BIC mechanism.
\end{theorem}

\subsection{Partial Correlations}

In some settings, there could be partial correlations between the distributions of the values of each bidder for the various goods.
Our construction and analysis can also be modified to analyze such settings, to obtain sample bounds that are polynomial in the independent dimensions. To give a few examples:
\begin{itemize}
\item
If there are correlations across values of different goods for the same bidder, but different bidders' valuations are independent, then our upper bound for $|M|$ becomes $(S+1)^{n\cdot(\lceil\nicefrac{H}{\varepsilon}\rceil)^m}$, and so our analysis would yield sample complexity that is polynomial in the number of bidders (as our upper bound of $|M|$ in this case is still singly exponential in the number of bidders) and exponential in the number of goods (as our upper bound of $|M|$ in this case is doubly exponential in the number of goods).
\item
If for each bidder $i$ the values of every two goods $2j,2j\!+\!1\in[m]$ are correlated, but are independent of the values of these goods for any other bidder, and of the values of any other good for any bidder, then our upper bound for $|M|$ becomes $(S+1)^{n\cdot\nicefrac{m}{2}\cdot(\lceil\nicefrac{H}{\varepsilon}\rceil)^2}$, and so our analysis still yields sample complexity that is polynomial in the both the number of bidders and the number of goods (as our upper bound of $|M|$ in this case is still singly exponential in both parameters). This is an example for a weaker form of correlation for which our analysis can still yield sample complexity that is polynomial in all parameters.
\end{itemize}

\section*{Acknowledgments}
\begin{sloppypar}
Gonczarowski was supported in part by the Adams Fellowship Program of the Israel Academy of Sciences and Humanities.
His work was supported in part by ISF grant 1435/14 administered by the Israeli Academy of Sciences, by Israel-USA Bi-national Science Foundation (BSF) grant number 2014389, and
by the European Research Council (ERC) under the European Union's Horizon 2020 research and innovation programme (grant agreement No 740282).
Weinberg was supported by NSF CCF-1717899 and by NSF CAREER award CCF-1942497.
We thank Noam Nisan and Costis Daskalakis for helpful conversations, and thank anonymous referees for helpful feedback.
\end{sloppypar}

\bibliographystyle{abbrvnat}
\bibliography{sampling}

\appendix

\section{Dominant-Strategy Incentive Compatibility}\label{dsic}

In this \lcnamecref{dsic}, we demonstrate a somewhat less involved version of our analysis and prove the following result, which we find interesting in its own right --- a version of our main result for dominant strategy (rather than Bayesian) incentive compatibility. The proofs of supporting \lcnamecrefs{dsic-opteps-opt} are relegated to \cref{proofs}. We start by phrasing a weak technical assumption that this result additionally requires.

\begin{definition}[Weakly Downward Closed]
We say that the set of allowed allocations $X$ is \emph{weakly downward closed} if for every outcome $x\in X$ and for every $k\in [n]$, there exists~$x'\in X$ such that \ \ a)~$v_k(x')=v_k(x)$ for every $v_k=(v_{k,j})_{j\in[m]}$, and\ \ b)~$v_i(x')=0$ for every $i\ne k$ and $v_i=(v_{i,j})_{j\in[m]}$.\footnote{For example, for the multi-item setting this is satisfied if for every $x=(x_{i,j})_{i\in[n],j\in[m]}\in X$ and for every $i\in [n]$, we have that $(x_i,0_{-i})\in X$ as well.}
\end{definition}

\begin{theorem}[Main Result --- DSIC Variant]\label{polynomial-dsic}
If $X$ is weakly downward closed, then for every $\varepsilon,\delta>0$, the sample complexity of learning an IR and $\varepsilon$-DSIC auction whose revenue differs from that of the optimal IR and DSIC auction by less than an additive $\varepsilon$ is at most $\poly(n,m,L,H,\nicefrac{1}{\varepsilon},\log\nicefrac{1}{\delta})$. That is, there exists a deterministic algorithm\footnote{Once  again, recall that this result is information-theoretic and not computationally efficient (by necessity, without resolving major open problems), so our decision maker (seller) is computationally unbounded, and we allow the algorithm to make calls to any deterministic oracle that has no access to any $V_{i,j}$. In particular, we assume access to an oracle that can solve the revenue maximization problem on any precisely given $V'_{i,j}$ of finite support.} that given $\poly(n,m,L,H,\nicefrac{1}{\varepsilon},\log\nicefrac{1}{\delta})$ samples from each $V_{i,j}$, with probability $1\!-\!\delta$ outputs an IR and $\varepsilon$-DSIC auction that attains from $\bigtimes_{i,j} V_{i,j}$ expected revenue at most an additive~$\varepsilon$ smaller than any IR and DSIC auction.
\end{theorem}

\begin{proof}
We assume that for every $i\in[n]$ and $j\in[m]$, we have~$S$ (to be determined later) independent samples $(v_{i,j}^s)_{s=1}^S$ from $V_{i,j}$. \cref{dsic-alg} presents our learning \lcnamecref{dsic-alg}, which is similar in nature to the one presented in \citet{dhp16} for certain single-parameter environments, however the analysis that we will use to show that it does not overfit the samples is completely different (even for single-parameter environments, where our analysis holds for arbitrary weakly downward closed allocation constraints
\begin{algorithm}[t]
\small
\DontPrintSemicolon
\SetAlgoHangIndent{3.65em}
\SetKwProg{Fn}{Function}{:}{}
\SetKwFunction{EmpOptDSIC}{EmpiricalOptimizeDSIC}
\SetKwFunction{OptOrcDSIC}{OptimizationOracleDSIC}

\Fn{\EmpOptDSIC{$H,X,\varepsilon,\delta,(v^s_{i,j})_{i\in[n],j\in[m]}^{s\in[S]}$}}{
\KwIn{For every $i\in[n],j\in[m]$, $(v^s_{i,j})_{s=1}$ is a sequence of $S=\tilde{O}\left(\tfrac{H^2}{\varepsilon^2}\cdot\bigl(\log\nicefrac{1}{\delta}+\tfrac{nmH}{\varepsilon}\bigr)\right)$ samples from $V_{i,j}$}
\KwOut{With probability $1\!-\!\delta$, an IR and $(4mL\varepsilon)$-DSIC mechanism defined over $[0,H]^{n\cdot m}$ with allocations in $X$, whose expected revenue from $\bigtimes_{i,j}V_{i,j}$ is up to an additive $4nmL\varepsilon$ smaller than that of any IR and DSIC mechanism defined over $\bigtimes_{i,j}\supp W_{i,j}$ with allocations in $X$}
\For{$i\in[n],j\in[m]$}{
\For{$s\in[S]$}{
$w^s_{i,j}\longleftarrow\epsfloor{v^s_{i,j}}$\;
}
$W_{i,j}\longleftarrow$ \emph{the uniform distribution over $(w^s_{i,j})_{s=1}^S$}\;
}
$\mu\longleftarrow\OptOrcDSIC{$H,X,\varepsilon,(W_{i,j})_{i\in[n],j\in[m]}$}$\tcp*{See definition below}
\KwRet{The mechanism that for input $(v_{i,j})_{i\in[n],j\in[m]}$ outputs $\mu\Bigl(\bigl(\epsfloor{v_{i,j}}\bigr)_{i\in[n],j\in[m]}\Bigr)$, modified to charge each bidder $mL\varepsilon$ less.}
}

\vspace{1em}

\Fn{\OptOrcDSIC{$H,X,\varepsilon,(W_{i,j})_{i\in[n],j\in[m]}$}}{
\KwIn{For every $i\in[n],j\in[m]$, $W_{i,j}$ is a distribution over $[0,H]_{\varepsilon}$}
\KwOut{An IR and $(2mL\varepsilon)$-DSIC mechanism defined over $[0,H]_{\varepsilon}^{n\cdot m}$ with allocations in $X$, which maximizes the expected revenue from $\bigtimes_{i,j}W_{i,j}$ among all IR and $(2mL\varepsilon)$-DSIC mechanisms defined over $\bigtimes_{i,j}\supp W_{i,j}$ with allocations in $X$}
$\mu\longleftarrow$ \emph{an IR and $(2mL\varepsilon)$-DSIC mechanism defined over $\bigtimes_{i,j}\supp W_{i,j}$ with allocations \linebreak \vphantom{~}\hspace{-.8em}in $X$, which maximizes the expected revenue from $\bigtimes_{i,j}W_{i,j}$ among all such \linebreak \vphantom{~}\hspace{-.8em}mechanisms}\;
\KwRet{The mechanism obtained by extending $\mu$ to be defined over all $(v_{i,j})_{i,j}\in[0,H]_{\varepsilon}^{n\cdot m}$ as follows: if $v_{k,j}\notin\supp W_{k,j}$ for precisely one bidder $k$ (and one or more parameters $j$), then an outcome is chosen such that bidder~$k$'s valuation and payment are the same as in $\mu(v'_k,v_{-k})$, for $v'_k=(v'_{k,j})_{j\in[m]}\in\supp\bigtimes_jW_{k,j}$ that maximizes $\expect{}{u_k\bigl(v_k,\mu(v'_k,v_{-k})\bigr)}$, and such that all other bidders' valuations of the outcome are zero irrespective of their valuations, and all other bidders' payments are zero. If $v_{i,j}\notin\supp W_{i,j}$ for more than one bidder $i$, then an outcome is chosen such that all bidders' valuations are zero irrespective of their valuations and all payments are zero.}
}

\caption{Empirical Multi-Parameter Up-to-$\varepsilon$ DSIC Revenue Maximization.}\label{dsic-alg}
\end{algorithm}%
--- note that this assumption guarantees that the modification to the allocation rule of the mechanism returned by \OptOrcDSIC still results in allocations in $X$).
We now analyze \cref{dsic-alg}. Note that for every $i,j,s$, we have that $w_{i,j}^s\sim\epsfloor{V_{i,j}}$ independently.

Let $\mathcal{V}$ be the set of all product distributions $\bigtimes_{i,j} W'_{i,j}$ where each $W'_{i,j}$ is the uniform distribution over some multiset of $S$ values from $[0,H]_{\varepsilon}$. Let $M$ be the set of all mechanisms of the form \OptOrcDSIC{$H,X,\varepsilon,\bigtimes_{i,j} W'_{i,j}$} for all $\bigtimes_{i,j} W'_{i,j} \in \mathcal{V}$.
At the heart of our analysis is the observation that $|\mathcal{V}|<(S+1)^{n\cdot m\cdot\lceil\nicefrac{H}{\varepsilon}\rceil}$. (Crucially, this expression has~$S$ only in the base and not in the exponent!) Indeed, for every $\bigtimes_{i,j} W'_{i,j} \in \mathcal{V}$, for every~$i,j$, and for every integer multiples of $\varepsilon$ in $[0,H]$ (there are $\lceil\nicefrac{H}{\varepsilon}\rceil$ many such values), the probability of this value in $W'_{i,j}$ can be any of the $S+1$ values $0,\nicefrac{1}{S},\ldots,1$. (The inequality on $|\mathcal{V}|$ is strict since, for example, not all probabilities can be $0$ simultaneously.)
Therefore, $|M|<(S+1)^{n\cdot m\cdot\lceil\nicefrac{H}{\varepsilon}\rceil}$. 

Let $\OPT$ be the IR and DSIC auction defined over $\bigtimes_{i,j}\supp V_{i,j}$ that maximizes the expected revenue (among such auctions) from $\bigtimes_{i,j}V_{i,j}$ (our learning \lcnamecref{dsic-alg} cannot hope to find $\OPT$, but in our analysis we may carefully reason about it, as it is nonetheless well defined; in particular, it does not depend on $S$). Let $\OPT_{\varepsilon}$ be the (randomized) mechanism
defined over $\bigtimes_{i,j}\epsfloor{V_{i,j}}$ as follows: let $w=(w_{i,j})_{i\in[n],j\in[m]}$ be an input valuation; for each $i,j$, independently draw $v_{i,j}\sim V_{i,j}|_{[w_{i,j},w_{i,j}+\varepsilon)}$; let $(x,p)=\OPT\bigl((v_{i,j})_{i\in[n],j\in[m]}\bigr)$; the allocation of $\OPT_{\varepsilon}(w)$ is $x$, and the payment of each bidder $i$ is $p_i-mL\varepsilon$.

\begin{lemma}\label{dsic-opteps-opt}
$\OPT_{\varepsilon}$ is an IR and $(2mL\varepsilon)$-DSIC mechanism whose expected revenue from $\bigtimes_{i,j}\epsfloor{V_{i,j}}$ is $nmL\varepsilon$ smaller than the expected revenue of $\OPT$ from $\bigtimes_{i,j}V_{i,j}$.
\end{lemma}

We will choose $S$ so that with probability at least $1\!-\!\delta$, it simultaneously holds for all mechanisms $\mu\in M \cup \{\OPT_{\varepsilon}\}$ that
\begin{equation}\label{dsic-concentrate}
\bigl|\Rev_{\bigtimes_{i,j}W_{i,j}}(\mu)-\Rev_{\bigtimes_{i,j}\epsfloor{V_{i,j}}}(\mu)\bigr|\le nmL\varepsilon.
\end{equation}

By \cref{concentration} (with $\ell\eqdef n\cdot m$, and note that any mechanism's revenue is bounded by $nmLH$), we have that for each mechanism $\mu\in M \cup \{\OPT_{\varepsilon}\}$ separately, \cref{dsic-concentrate} holds with probability at least $1-\frac{4He^{-\varepsilon^2S/(8H^2)}}{\varepsilon}$. Choosing $S$ so that this probability is at least $1-\frac{\delta}{|M|+1}$, we obtain that \cref{dsic-concentrate} holds simultaneously for all mechanisms $\mu\in M \cup \{\OPT_{\varepsilon}\}$ with probability at least $1\!-\!\delta$. We now estimate $S$. Since $|M\cup\{OPT_{\varepsilon}\}|\leq(S+1)^{n\cdot m\cdot\lceil\nicefrac{H}{\varepsilon}\rceil}$, we have that it is enough to take $S$ such that
\[
S\ge\tfrac{8H^2}{\varepsilon^2}\cdot\bigl(\log\tfrac{4H}{\varepsilon}+\log\nicefrac{1}{\delta}+nm\lceil\nicefrac{H}{\varepsilon}\rceil\log(S+1)\bigr).
\]
Therefore,\footnote{See \cref{recursion-solution}.} there exists an appropriate
\[
S=O\left(\tfrac{H^2}{\varepsilon^2}\log\nicefrac{1}{\delta}+\tfrac{nmH^3}{\varepsilon^3}\log\tfrac{nmH^3}{\varepsilon^3}\right)=\tilde{O}\left(\tfrac{H^2}{\varepsilon^2}\cdot\bigl(\log\nicefrac{1}{\delta}+\tfrac{nmH}{\varepsilon}\bigr)\right).
\]

Let $\mu$ be the output of the call to $\OptOrcDSIC$ in \cref{dsic-alg} (note that it is well defined since $X$ is weakly downward closed), and let $\mu^{\varepsilon}$ be the final output of the \lcnamecref{dsic-alg} (the output of $\EmpOptDSIC$). 

\begin{lemma}\label{dsic-mueps-mu}
$\mu^{\varepsilon}$ is an IR and $(4mL\varepsilon$)-DSIC mechanism whose expected revenue from $\bigtimes_{i,j}V_{i,j}$ is $nmL\varepsilon$ smaller than the expected revenue of $\mu$ from $\bigtimes_{i,j}\epsfloor{V_{i,j}}$.
\end{lemma}

So, we have that $\mu^{\varepsilon}$ is an IR and $(4m\varepsilon)$-DSIC mechanism (by \cref{dsic-mueps-mu}) and that with probability at least $1\!-\!\delta$:
\begin{align*}
&\Rev_{\bigtimes_{i,j}V_{i,j}}(\mu^{\varepsilon})=\tag{by \cref{dsic-mueps-mu}}\\
=&\Rev_{\bigtimes_{i,j}\epsfloor{V_{i,j}}}(\mu)-nmL\varepsilon\ge\tag{by \cref{dsic-concentrate} and since $\mu\in M$}\\
\ge&\Rev_{\bigtimes_{i,j}W_{i,j}}(\mu)-2nmL\varepsilon\tag{since $\OPT_{\varepsilon}$ is IR and $(2m\varepsilon)$-DSIC and by optimality of $\mu$}\ge\\
\ge&\Rev_{\bigtimes_{i,j}W_{i,j}}(\OPT_{\varepsilon})-2nmL\varepsilon\tag{by \cref{dsic-concentrate} for $\OPT_{\varepsilon}$}\ge\\
\ge&\Rev_{\bigtimes_{i,j}\epsfloor{V_{i,j}}}(\OPT_{\varepsilon})-3nmL\varepsilon=\tag{by \cref{dsic-opteps-opt}}\\
=&\Rev_{\bigtimes_{i,j}V_{i,j}}(\OPT)-4nmL\varepsilon.\tag*{\qedhere}
\end{align*}
\end{proof}

Two main differences between the statements of \cref{polynomial-dsic} and of \cref{polynomial-bic} stand out: the technical requirement that $X$ be weakly downward closed, and the fact that the learned mechanism has weaker incentive properties than the benchmark. In \cref{bic} we have avoided the latter within the context of Bayesian incentive compatibility via \cref{eps-bic-bic}. It would be interesting to understand whether and to what extent a similar result also applies within the context of dominant-strategy incentive compatibility.

\begin{open-problem}\label{eps-dsic-dsic}
For $n\!>\!1$ bidders and $m\!>\!1$ items, even with additive bidders, is the maximum revenue attainable by any IR and $\varepsilon$-DSIC auction from a given product distribution at most negligibly ($\poly(\varepsilon)\!\cdot\!\poly(n,m,H)$) greater than the maximum revenue attainable by any IR and DSIC auction from the same distribution?
\end{open-problem}

Since the $\varepsilon$-dominant strategy incentive compatibility implies $\varepsilon$-Bayesian incentive compatibility, then the question in \cref{eps-dsic-dsic} holds whenever the optimal BIC auction is also DSIC. Recently, \citet{y17} has shown that in general the revenue obtained by the optimal IR and BIC auctions can be higher than that obtained by the optimal IR and DSIC auctions --- it is such cases in which \cref{eps-dsic-dsic} is open and interesting. An affirmative result for \cref{eps-dsic-dsic} would of course immediately imply an analogue of \cref{eps-bic-bic-samples} for DSIC auctions, which, if true, would allow \cref{polynomial-dsic} to be strengthened to allow for learning a precisely DSIC auction.

\begin{open-problem}
Given an IR and $\varepsilon$-DSIC auction and given some product distribution over $n\!>\!1$ bidders and $m\!>\!1$ items, even with additive bidders, is it possible to transform the given auction into an IR and (precisely) DSIC auction with negligible ($\poly(\varepsilon)\!\cdot\!\poly(n,m,H)$) revenue loss from that distribution using polynomially many samples from this product distribution?
\end{open-problem}

\section{Omitted Proofs}\label{proofs}

\subsection{Proofs Omitted from Section~\ref{bic}}

\begin{proof}[Proof of \cref{bic-opteps-opt}]

Individual Rationality: Let $k\in[n]$, let $w=(w_{i,j})_{i\in[n],j\in[m]}$ be a (rounded) valuation profile, and let $v=(v_{i,j})_{i\in[n],j\in[m]}$ be drawn as defined above. Then:
\begin{align*}
&\expect{}{u_k\bigl(w_k,\OPT_{\varepsilon}(w)\bigr)}=\\
=\;&\expect{}{u_k\bigl(w_k,\OPT(v)\bigr)+mL\varepsilon}\ge\tag{since valuations are Lipschitz}\\
\ge\;&\expect{}{u_k\bigl(v_k,\OPT(v)\bigr)-\sum_{j\in[m]}L(v_{k,j}-w_{k,j})+mL\varepsilon}\ge\tag{since $v_{k,j}<w_{k,j}+\varepsilon$}\\
\ge\;&\expect{}{u_k\bigl(v_k,\OPT(v)\bigr)}\ge0.\tag{since $\OPT$ is IR}
\end{align*}

Incentive Compatibility: Let $k\in[n]$, let $w_k=(w_{k,j})_{j\in[m]}$ and $w'_k=(w'_{k,j})_{j\in[m]}$ be (rounded) valuations, and let $v_k=(v_{k,j})_{j\in[m]}$ and $v'_k=(v'_{k,j})_{j\in[m]}$ be respectively drawn as defined above. Then:
\begin{align*}
&\expect{w_{-k}\sim\bigtimes_{\substack{i,j\\i\ne k}}\epsfloor{V_{i,j}}}{u_k\bigl(w_k,\OPT_{\varepsilon}(w_k,w_{-k})\bigr)}=\\
=\;&\expect{v_{-k}\sim\bigtimes_{\substack{i,j\\i\ne k}}V_{i,j}}{u_k\bigl(w_k,\OPT(v_k,v_{-k})\bigr)+mL\varepsilon}\ge\tag{since valuations are Lipschitz}\\
\ge\;&\expect{v_{-k}\sim\bigtimes_{\substack{i,j\\i\ne k}}V_{i,j}}{u_k\bigl(v_k,\OPT(v_k,v_{-k})\bigr)-\sum_{j\in[m]}L(v_{k,j}-w_{k,j})+mL\varepsilon}\ge\tag{since $v_{k,j}<w_{k,j}+\varepsilon$}\\
\ge\;&\expect{v_{-k}\sim\bigtimes_{\substack{i,j\\i\ne k}}V_{i,j}}{u_k\bigl(v_k,\OPT(v_k,v_{-k})\bigr)}\ge\tag{since $\OPT$ is $(4mL\varepsilon)$-BIC for $\bigtimes_{i,j}V_{i,j}$}\\
\ge\;&\expect{v_{-k}\sim\bigtimes_{\substack{i,j\\i\ne k}}V_{i,j}}{u_k\bigl(v_k,\OPT(v'_k,v_{-k})\bigr)}-4mL\varepsilon\ge\tag{since valuations are Lipschitz}\\
\ge\;&\expect{v_{-k}\sim\bigtimes_{\substack{i,j\\i\ne k}}V_{i,j}}{u_k\bigl(w_k,\OPT(v'_k,v_{-k})\bigr)-\sum_{j\in[m]}L(v_{k,j}-w_{k,j})}-4mL\varepsilon\ge\tag{since $v_{k,j}<w_{k,j}+\varepsilon$}\\
\ge\;&\expect{v_{-k}\sim\bigtimes_{\substack{i,j\\i\ne k}}V_{i,j}}{u_k\bigl(w_k,\OPT(v'_k,v_{-k})\bigr)}-5mL\varepsilon=\\
=\;&\expect{w_{-k}\sim\bigtimes_{\substack{i,j\\i\ne k}}\epsfloor{V_{i,j}}}{u_k\bigl(w_k,\OPT_{\varepsilon}(w'_k,w_{-k})\bigr)-mL\varepsilon}-5mL\varepsilon=\\
=\;&\expect{w_{-k}\sim\bigtimes_{\substack{i,j\\i\ne k}}\epsfloor{V_{i,j}}}{u_k\bigl(w_k,\OPT_{\varepsilon}(w'_k,w_{-k})\bigr)}-6mL\varepsilon.
\end{align*}

Revenue: Let $OPT_{\varepsilon}=(x_{\varepsilon},p_{\varepsilon})$ and let $\OPT=(x,p)$. Then:
\begin{align*}
\Rev_{\bigtimes_{i,j}\epsfloor{V_{i,j}}}(\OPT_{\varepsilon})=\;
&\expect{w\sim\bigtimes_{i,j}\epsfloor{V_{i,j}}}{p_{\varepsilon}(w)}=\\
=\;&\expect{v\sim\bigtimes_{i,j}V_{i,j}}{p(v)-\sum_{i\in[n]}mL\varepsilon}=\\
=\;&\expect{v\sim\bigtimes_{i,j}V_{i,j}}{p(v)}-nmL\varepsilon=\\
=\;&\Rev_{\bigtimes_{i,j}V_{i,j}}(\OPT)-nmL\varepsilon.\qedhere
\end{align*}
\end{proof}

\begin{proof}[Proof of \cref{bic-mu-opteps}]
Incentive Compatibility of $\OPT_{\varepsilon}$: Let $k\in[n]$, let $w_k=(w_{k,j})_{j\in[m]}$ and $w'_k=(w'_{k,j})_{j\in[m]}$ be (rounded) valuations. Then:
\begin{align*}
&\expect{w_{-k}\sim\bigtimes_{\substack{i,j\\i\ne k}}W_{i,j}}{u_k\bigl(w_k,\OPT_{\varepsilon}(w_k,w_{-k})\bigr)}\ge\tag{by \cref{bic-ic-concentrate} for $\OPT_{\varepsilon}$}\\
\ge\;&\expect{w_{-k}\sim\bigtimes_{\substack{i,j\\i\ne k}}\epsfloor{V_{i,j}}}{u_k\bigl(w_k,\OPT_{\varepsilon}(w_k,w_{-k})\bigr)}-mL\varepsilon\ge\tag{since $\OPT_{\varepsilon}$ is $(6mL\varepsilon)$-BIC for $\bigtimes_{i,j}\epsfloor{V_{i,j}}$}\\
\ge\;&\expect{w_{-k}\sim\bigtimes_{\substack{i,j\\i\ne k}}\epsfloor{V_{i,j}}}{u_k\bigl(w_k,\OPT_{\varepsilon}(w'_k,w_{-k})\bigr)}-7mL\varepsilon\ge\tag{by \cref{bic-ic-concentrate} for $\OPT_{\varepsilon}$}\\
\ge\;&\expect{w_{-k}\sim\bigtimes_{\substack{i,j\\i\ne k}}W_{i,j}}{u_k\bigl(w_k,\OPT_{\varepsilon}(w'_k,w_{-k})\bigr)}-8mL\varepsilon.
\end{align*}

Revenue: by \cref{eps-bic-bic} (for $\bigtimes_{i,j}W_{i,j}$), since $\OPT_{\varepsilon}$ is IR and $(8mL\varepsilon)$-BIC for $\bigtimes_{i,j}W_{i,j}$ and by optimality of $\mu$.

Incentive Compatibility of $\mu$: Let $k\in[n]$, let $w_k=(w_{k,j})_{j\in[m]}$ and $w'_k=(w'_{k,j})_{j\in[m]}$ be (rounded) valuations. Then:
\begin{align*}
&\expect{w_{-k}\sim\bigtimes_{\substack{i,j\\i\ne k}}\epsfloor{V_{i,j}}}{u_k\bigl(w_k,\mu(w_k,w_{-k})\bigr)}\ge\tag{by \cref{bic-ic-concentrate} and since $\mu\in M$}\\
\ge\;&\expect{w_{-k}\sim\bigtimes_{\substack{i,j\\i\ne k}}W_{i,j}}{u_k\bigl(w_k,\mu(w_k,w_{-k})\bigr)}-mL\varepsilon\ge\tag{since $\mu$ is BIC for $\bigtimes_{i,j}W_{i,j}$}\\
\ge\;&\expect{w_{-k}\sim\bigtimes_{\substack{i,j\\i\ne k}}W_{i,j}}{u_k\bigl(w_k,\mu(w'_k,w_{-k})\bigr)}-mL\varepsilon\ge\tag{by \cref{bic-ic-concentrate} since $\mu\in M$}\\
\ge\;&\expect{w_{-k}\sim\bigtimes_{\substack{i,j\\i\ne k}}\epsfloor{V_{i,j}}}{u_k\bigl(w_k,\mu(w'_k,w_{-k})\bigr)}-2mL\varepsilon.\tag*{\qedhere}
\end{align*}
\end{proof}

\begin{proof}[Proof of \cref{bic-mueps-mu}]
Individual Rationality: let $k\in[n]$ and let $v=(v_{i,j})_{i\in[n],j\in[m]}$ be a valuation profile. Then:
\begin{align*}
&\expect{}{u_k\bigl(v_k,\mu^{\varepsilon}(v)\bigr)}=\\
=\;&\expect{}{u_k\bigl(v_k,\mu\bigl((\epsfloor{v_{i,j}})_{i\in[n],j\in[m]}\bigr)\bigr)+mL\varepsilon}\ge\tag{since valuations are Lipschitz}\\
\ge\;&\expect{}{u_k\bigl((\epsfloor{v_{k,j}})_{j\in[m]},\mu\bigl((\epsfloor{v_{i,j}})_{i\in[n],j\in[m]}\bigr)\bigr)-\sum_{j\in[m]}L(v_{k,j}-\epsfloor{v_{k,j}})+mL\varepsilon}\ge\\
\ge\;&\expect{}{u_k\bigl((\epsfloor{v_{k,j}})_{j\in[m]},\mu\bigl((\epsfloor{v_{i,j}})_{i\in[n],j\in[m]}\bigr)\bigr)}\ge0.\tag{since $\mu$ is IR}
\end{align*}

Revenue: Let $\mu^{\varepsilon}=(x^{\varepsilon},p^{\varepsilon})$ and let $\mu=(x,p)$. Then:
\begin{align*}
\Rev_{\bigtimes_{i,j}V_{i,j}}(\mu^{\varepsilon})=\;
&\expect{v\sim\bigtimes_{i,j}V_{i,j}}{p^{\varepsilon}(v)}=\\
=\;&\expect{w\sim\bigtimes_{i,j}\epsfloor{V_{i,j}}}{p(w)-\sum_{i\in[n]}mL\varepsilon}=\\
=\;&\expect{w\sim\bigtimes_{i,j}\epsfloor{V_{i,j}}}{p(w)}-nmL\varepsilon=\\
=\;&\Rev_{\bigtimes_{i,j}\epsfloor{V_{i,j}}}(\mu)-nmL\varepsilon.
\end{align*}

Incentive Compatibility: Let $k\in[n]$, let $v_k=(v_{k,j})_{j\in[m]}$ and $v'_k=(v'_{k,j})_{j\in[m]}$ be valuations, and denote
$\epsfloor{v_k}=(\epsfloor{v_{k,j}})_{j\in[m]}$ and $\epsfloor{v'_k}=(\epsfloor{v'_{k,j}})_{j\in[m]}$. Then:
\begin{align*}
&\expect{v_{-k}\sim\bigtimes_{\substack{i,j\\i\ne k}}V_{i,j}}{u_k\bigl(v_k,\mu^{\varepsilon}(v_k,v_{-k})\bigr)}=\\
=\;&\expect{w_{-k}\sim\bigtimes_{\substack{i,j\\i\ne k}}\epsfloor{V_{i,j}}}{u_k\bigl(v_k,\mu(\epsfloor{v_k},w_{-k})\bigr)+mL\varepsilon}\ge\tag{since valuations are Lipschitz}\\
\ge\;&\expect{w_{-k}\sim\bigtimes_{\substack{i,j\\i\ne k}}\epsfloor{V_{i,j}}}{u_k\bigl(\epsfloor{v_k},\mu(\epsfloor{v_k},w_{-k})\bigr)-\sum_{j\in[m]}L\bigl(v_{k,j}-\epsfloor{v_{k,j}}\bigr)+mL\varepsilon}\ge\\
\ge\;&\expect{w_{-k}\sim\bigtimes_{\substack{i,j\\i\ne k}}\epsfloor{V_{i,j}}}{u_k\bigl(\epsfloor{v_k},\mu(\epsfloor{v_k},w_{-k})\bigr)}\ge\tag{since $\mu$ is $(2mL\varepsilon)$-BIC for $\bigtimes_{i,j}\epsfloor{V_{i,j}}$}\\
\ge\;&\expect{w_{-k}\sim\bigtimes_{\substack{i,j\\i\ne k}}\epsfloor{V_{i,j}}}{u_k\bigl(\epsfloor{v_k},\mu(\epsfloor{v'_k},w_{-k})\bigr)}-2mL\varepsilon\ge\tag{since valuations are Lipschitz}\\
\ge\;&\expect{w_{-k}\sim\bigtimes_{\substack{i,j\\i\ne k}}\epsfloor{V_{i,j}}}{u_k\bigl(v_k,\mu(\epsfloor{v'_k},w_{-k})\bigr)-\sum_{j\in[m]}L\bigl(v_{k,j}-\epsfloor{v_{k,j}}\bigr)}-2mL\varepsilon\ge\\
\ge\;&\expect{w_{-k}\sim\bigtimes_{\substack{i,j\\i\ne k}}\epsfloor{V_{i,j}}}{u_k\bigl(v_k,\mu(\epsfloor{v'_k},w_{-k})\bigr)}-3mL\varepsilon=\\
=\;&\expect{v_{-k}\sim\bigtimes_{\substack{i,j\\i\ne k}}V_{i,j}}{u_k\bigl(v_k,\mu^{\varepsilon}\bigl(v_k,v_{-k}\bigr)}-4mL\varepsilon.\tag*{\qedhere}
\end{align*}
\end{proof}

\subsection{Proofs Omitted from Appendix \ref{dsic}}

\begin{proof}[Proof of \cref{dsic-opteps-opt}]
Individual Rationality: Let $k\in[n]$, let $w=(w_{i,j})_{i\in[n],j\in[m]}$ be a (rounded) valuation profile, and let $v=(v_{i,j})_{i\in[n],j\in[m]}$ be drawn as defined above. Then:
\begin{align*}
&\expect{}{u_k\bigl(w_k,\OPT_{\varepsilon}(w)\bigr)}=\\
=\;&\expect{}{u_k\bigl(w_k,\OPT(v)\bigr)+mL\varepsilon}\ge\tag{since valuations are Lipschitz}\\
\ge\;&\expect{}{u_k\bigl(v_k,\OPT(v)\bigr)-\sum_{j\in[m]}L(v_{k,j}-w_{k,j})+mL\varepsilon}\ge\tag{since $v_{k,j}<w_{k,j}+\varepsilon$}\\
\ge\;&\expect{}{u_k\bigl(v_k,\OPT(v)\bigr)}\ge0.\tag{since $\OPT$ is IR}
\end{align*}

Incentive Compatibility: Let $k\in[n]$, let $w=(w_{i,j})_{i\in[n],j\in[m]}$ be a (rounded) valuation profile, let $w'_k=(w'_{k,j})_{j\in[m]}$ be a (rounded) valuation, and let $v=(v_{i,j})_{i\in[n],j\in[m]}$ and $v'_k=(v'_{k,j})_{j\in[m]}$ be respectively drawn as defined above. Then:
\begin{align*}
&\expect{}{u_k\bigl(w_k,\OPT_{\varepsilon}(w_k,w_{-k})\bigr)}=\\
=\;&\expect{}{u_k\bigl(w_k,\OPT(v_k,v_{-k})\bigr)+mL\varepsilon}\ge\tag{since valuations are Lipschitz}\\
\ge\;&\expect{}{u_k\bigl(v_k,\OPT(v_k,v_{-k})\bigr)-\sum_{j\in[m]}L(v_{k,j}-w_{k,j})+mL\varepsilon}\ge\tag{since $v_{k,j}<w_{k,j}+\varepsilon$}\\
\ge\;&\expect{}{u_k\bigl(v_k,\OPT(v_k,v_{-k})\bigr)}\ge\tag{since $\OPT$ is DSIC}\\
\ge\;&\expect{}{u_k\bigl(v_k,\OPT(v'_k,v_{-k})\bigr)}\ge\tag{since valuations are Lipschitz}\\
\ge\;&\expect{}{u_k\bigl(w_k,\OPT(v'_k,v_{-k})\bigr)-\sum_{j\in[m]}L(v_{k,j}-w_{k,j})}\ge\tag{since $v_{k,j}<w_{k,j}+\varepsilon$}\\
\ge\;&\expect{}{u_k\bigl(w_k,\OPT(v'_k,v_{-k})\bigr)}-mL\varepsilon=\\
=\;&\expect{}{u_k\bigl(w_k,\OPT_{\varepsilon}(w'_k,w_{-k})\bigr)-mL\varepsilon}-mL\varepsilon=\\
=\;&\expect{}{u_k\bigl(w_k,\OPT_{\varepsilon}(w'_k,w_{-k})\bigr)}-2mL\varepsilon.
\end{align*}

Revenue: Let $OPT_{\varepsilon}=(x_{\varepsilon},p_{\varepsilon})$ and let $\OPT=(x,p)$. Then:
\begin{align*}
\Rev_{\bigtimes_{i,j}\epsfloor{V_{i,j}}}(\OPT_{\varepsilon})=\;
&\expect{w\sim\bigtimes_{i,j}\epsfloor{V_{i,j}}}{p_{\varepsilon}(w)}=\\
=\;&\expect{v\sim\bigtimes_{i,j}V_{i,j}}{p(v)-\sum_{i\in[n]}mL\varepsilon}=\\
=\;&\expect{v\sim\bigtimes_{i,j}V_{i,j}}{p(v)}-nmL\varepsilon=\\
=\;&\Rev_{\bigtimes_{i,j}V_{i,j}}(\OPT)-nmL\varepsilon.\qedhere
\end{align*}
\end{proof}

\begin{proof}[Proof of \cref{dsic-mueps-mu}]
Individual Rationality: let $k\in[n]$ and let $v=(v_{i,j})_{i\in[n],j\in[m]}$ be a valuation profile. Then:
\begin{align*}
&\expect{}{u_k\bigl(v_k,\mu^{\varepsilon}(v)\bigr)}=\\
=\;&\expect{}{u_k\bigl(v_k,\mu\bigl((\epsfloor{v_{i,j}})_{i\in[n],j\in[m]}\bigr)\bigr)+mL\varepsilon}\ge\tag{since valuations are Lipschitz}\\
\ge\;&\expect{}{u_k\bigl((\epsfloor{v_{k,j}})_{j\in[m]},\mu\bigl((\epsfloor{v_{i,j}})_{i\in[n],j\in[m]}\bigr)\bigr)-\sum_{j\in[m]}L(v_{k,j}-\epsfloor{v_{k,j}})+mL\varepsilon}\ge\\
\ge\;&\expect{}{u_k\bigl((\epsfloor{v_{k,j}})_{j\in[m]},\mu\bigl((\epsfloor{v_{i,j}})_{i\in[n],j\in[m]}\bigr)\bigr)}\ge0.\tag{since $\mu$ is IR}
\end{align*}

Incentive Compatibility: Let $k\in[n]$, let $v=(v_{i,j})_{i\in[n],j\in[m]}$ be a valuation profile, let $v'_k=(v'_{k,j})_{j\in[m]}$ be a valuation, and denote
$\epsfloor{v_k}=(\epsfloor{v_{k,j}})_{j\in[m]}$, $\epsfloor{v'_k}=(\epsfloor{v'_{k,j}})_{j\in[m]}$ and $\epsfloor{v_{-k}}=(\epsfloor{v_{i,j}})_{i\in[n]\setminus\{k\},j\in[m]}$. Then:
\begin{align*}
&\expect{}{u_k\bigl(v_k,\mu^{\varepsilon}(v_k,v_{-k})\bigr)}=\\
=\;&\expect{}{u_k\bigl(v_k,\mu(\epsfloor{v_k},\epsfloor{v_{-k}})\bigr)+mL\varepsilon}\ge\tag{since valuations are Lipschitz}\\
\ge\;&\expect{}{u_k\bigl(\epsfloor{v_k},\mu(\epsfloor{v_k},\epsfloor{v_{-k}})\bigr)-\sum_{j\in[m]}L\bigl(v_{k,j}-\epsfloor{v_{k,j}}\bigr)+mL\varepsilon}\ge\\
\ge\;&\expect{}{u_k\bigl(\epsfloor{v_k},\mu(\epsfloor{v_k},\epsfloor{v_{-k}})\bigr)}\ge\tag{since $\mu$ is $(2mL\varepsilon)$-DSIC}\\
\ge\;&\expect{}{u_k\bigl(\epsfloor{v_k},\mu(\epsfloor{v'_k},\epsfloor{v_{-k}})\bigr)}-2mL\varepsilon\ge\tag{since valuations are Lipschitz}\\
\ge\;&\expect{}{u_k\bigl(v_k,\mu(\epsfloor{v'_k},\epsfloor{v_{-k}})\bigr)-\sum_{j\in[m]}L\bigl(v_{k,j}-\epsfloor{v_{k,j}}\bigr)}-2mL\varepsilon\ge\\
\ge\;&\expect{}{u_k\bigl(v_k,\mu(\epsfloor{v'_k},\epsfloor{v_{-k}})\bigr)}-3mL\varepsilon=\\
=\;&\expect{}{u_k\bigl(v_k,\mu^{\varepsilon}\bigl(v'_k,v_{-k}\bigr)}-4mL\varepsilon.
\end{align*}

Revenue: Let $\mu^{\varepsilon}=(x^{\varepsilon},p^{\varepsilon})$ and let $\mu=(x,p)$. Then:
\begin{align*}
\Rev_{\bigtimes_{i,j}V_{i,j}}(\mu^{\varepsilon})=\;
&\expect{v\sim\bigtimes_{i,j}V_{i,j}}{p^{\varepsilon}(v)}=\\
=\;&\expect{w\sim\bigtimes_{i,j}\epsfloor{V_{i,j}}}{p(w)-\sum_{i\in[n]}mL\varepsilon}=\\
=\;&\expect{w\sim\bigtimes_{i,j}\epsfloor{V_{i,j}}}{p(w)}-nmL\varepsilon=\\
=\;&\Rev_{\bigtimes_{i,j}\epsfloor{V_{i,j}}}(\mu)-nmL\varepsilon.\tag*{\qedhere}
\end{align*}
\end{proof}

\end{document}